\documentclass[12pt]{article}

\usepackage[affil-it]{authblk}

\usepackage{graphicx}
\usepackage{color}
\usepackage{array}
\usepackage{fullpage}
\usepackage{fancyhdr}
\usepackage{amsmath,amsthm}
\usepackage{amssymb,latexsym}
\usepackage{mathrsfs}
\usepackage{cancel}
\usepackage [pdfstartview={FitH},colorlinks]
{hyperref}
\usepackage [initials,short-journals,short-months]
{amsrefs}
\usepackage{url}
\usepackage{tikz}
\usepackage{enumitem}
\usepackage{soul}
\usepackage{bbm}
\usepackage[margin=.9in]{geometry}

\newtheorem{theorem}{Theorem}[section]
\newtheorem{lemma}[theorem]{Lemma}
\newtheorem{proposition}[theorem]{Proposition}

\theoremstyle{definition}

\newtheorem{assumption}[theorem]{Assumption}

\theoremstyle{remark}
\newtheorem{remark}[theorem]{Remark}
\newtheorem{example}[theorem]{Example}

\def\XXint#1#2#3{{
\setbox0=\hbox{$#1{#2#3}{\int}$}
\vcenter{\hbox{$#2#3$}}\kern-.5\wd0}}

\def \CSmooth(#1,#2){\mathcal{C}_{#1,#2}}

\def \Mgale(#1,#2){M_{#1}^{#2}}
\def \Ngale(#1,#2){\mathcal{N}_{#1}^{#2}}

\title{Spectral dimension and Bohr's formula for Schr\"odinger operators on unbounded fractal spaces}
\author{
Joe P. Chen\footnote{Department of Mathematics, University of Connecticut, Storrs CT 06269. \\Research supported in part by NSF grant DMS-1106982.\\\href{mailto:joe.p.chen@uconn.edu}
{joe.p.chen@uconn.edu}
\url
{http://homepages.uconn.edu/jpchen/}\\
\href{mailto:teplyaev@member.ams.org}
{teplyaev@member.ams.org}
\url
{http://homepages.uconn.edu/teplyaev/}
}
\\ \small{University of Connecticut}
\and
Stanislav Molchanov\footnote{Dept of Mathematics and Statistics, UNC at Charlotte, NC 28223.
\\Research supported in part by NSF grant DMS-1008132.\\\href{mailto:smolchan@uncc.edu}
{smolchan@uncc.edu} 
\url
{http://math.uncc.edu/people/stanislav-molchanov}
}
\\ \small{University of North Carolina-Charlotte \\and\\ The Russian Federation National Research University Higher School of Economics}
\and
Alexander Teplyaev$^*$
\\ \small{University of Connecticut}
}

\date{\today}

\numberwithin{equation}{section}

\begin{document}

\maketitle

\abstract{
We establish  an asymptotic formula for the eigenvalue counting function of the Schr\"odinger operator $-\Delta +V$ for some unbounded potentials $V$ on several types of unbounded fractal spaces. We give   sufficient conditions for Bohr's formula to hold on metric measure spaces which admit a cellular decomposition, and then verify these conditions for fractafolds and fractal fields based on nested fractals. In particular, we  partially answer a question of Fan, Khandker, and Strichartz regarding the spectral asymptotics of the harmonic oscillator potential on the infinite blow-up of a Sierpinski gasket. 
\newpage
\tableofcontents 
}

\section{Introduction} \label{sec:Intro}

In this paper we present
an asymptotic formula for the eigenvalue counting function of the Schr\"odinger operator $-\Delta +V$ for unbounded potentials $V$ on several types of unbounded fractal spaces. Such an asymptotic formula is often attributed to Niels Bohr in the Euclidean setting. We identify a set of sufficient conditions for Bohr's formula to hold on locally self-similar metric measure spaces which admit a cellular decomposition, and then verify these conditions for fractafolds \cites{StrFractafold,StrTep} and fractal fields \cite{HamKumFields} based on nested fractals. 
In particular, we are able to partially answer a question of Fan, Khandker, and Strichartz 
\cite{StrichartzSHO}
 regarding the spectral asymptotics of the harmonic oscillator potential on the infinite blow-up of a Sierpinski gasket (abbreviated $SG$). 
 
 All these results have similarities in the classical theory of 1D Sturm-Liouville operators (see \cite{ReedSimonVol4}). The deep analogy between nested fractals (the typical representative being $SG$) and the real line $\mathbb{R}_+^1 = [0,\infty)$ is related to the fact that all of them are \emph{finitely ramified}. (A set is said to be \emph{finitely ramified} if it can be divided into several disconnected subsets upon removing a finite number of points from the set. For $\mathbb{R}_+^1$ it suffices to remove one point; for $SG$, two points.)

Let us recall several known results from the spectral theory of the 1D Sch{\"o}dinger operator
\begin{equation}
H\psi = -\psi'' + V(x)\psi, \qquad x\geq 0
\end{equation}
with boundary condition at $x=0$ of either Dirichlet type, $\psi(0)=0$, or Neumann type, $\psi'(0)=0$.

\begin{enumerate}[label=\Roman*.]
\item 
\label{iBohr}
Assume that $V(x) \to +\infty$ as $|x|\to +\infty$. Then, by the result of H. Weyl,  the spectrum of $H$ in $L^2([0,\infty), dx)$ is discrete and, under some technical conditions,
\begin{equation}
N(\lambda,V) := \#\{\lambda_i(H)\leq \lambda\} \sim \frac{1}{\pi}\int_0^\infty \sqrt{(\lambda-V(x))_+} \,dx.
\end{equation}
This is known as N. Bohr's formula, see \cites{LevitanSargsjan,KS,HoltMolchanov}. 

\item \label{iII} Assume that $V(x)$ is compactly supported, or (weaker assumption) 
vanishing
fast enough (see below). Put $V(x) = V_+(x) - V_-(x)$, where $V_+=\max(0,V)$ and $V_-=\max(0,-V)$, and \begin{equation}N_-(V):=\#\{\lambda_i \leq 0\}   \leq N_-(-V_-(\cdot)).\end{equation}
The estimate of $N_-(V)$ as a result can be reduced to the negative potentials (potential wells). We   use   the notation $N_-(V)$ assuming 
here 
that $V(x) = -V_-(x)\leq 0$. 
The following estimates of $N_-(V)$ are popular in applications (see \cite{ReedSimonVol4}):
\begin{enumerate}
\item\label{iBa}(Bargmann)
\begin{equation}
N_-(V) \leq 1+\int_0^\infty x V(x)\,dx.
\end{equation}
\item\label{iBb}(Calogero) If $V(x)$ decreases with $|x|$ as $|x|\to\infty$, then
\begin{equation}
N_-(V) \leq c_0 \int_0^\infty \sqrt{V(x)}\,dx.
\end{equation}
The Calogero estimate has the correct scaling in the following sense.
\item\label{iBc} Consider the operator
\begin{equation}
H_\sigma \psi = -\psi'' + \sigma V_0(x) \psi,\qquad x\geq 0 \quad (\text{plus boundary condition}).
\end{equation}
Then as $\sigma\to\infty$,
\begin{equation}
N_-(\sigma V_0) \sim c_1 \sigma^{1/2} \int_0^\infty \sqrt{V(x)}\,dx.
\end{equation}
This is the so-called \emph{quasiclassical asymptotics}. 
It is an important problem to find such an estimate for $N_-(V)$, which has true scaling in $\mathbb{R}^d$, $d\geq 2$, \emph{i.e.,} for any $\sigma$,
\begin{equation}
N_-(\sigma V_0) \leq \sigma^{d/2} \Phi(V_0). \qquad (\text{Cwickl-Lieb-Rosenblum})
 \end{equation} 
For $d\geq 3$ this is the CLR estimate
\begin{equation}
N_-(V) \leq c_d \int\limits_{\mathbb{R}^d} |V(x)|^{d/2}\,dx.
\end{equation}
For $d=2$ the recent results by Grigor'yan and Nadirashivili \cite{GriNad1} and Shargorodsky \cite{Shargorodsky} give the desirable (though not simple) estimate. The paper \cite{GriNad2} contains the justification of the physical conjecture by Madau and Wu on $N_-(V)$ for 2D operators.
The case $d=1$ was studied in the relatively recent papers by K. Naimark, G. Rozenblum,   M. Solomyak et al (see \cites{NaimarkSolomyak,RozenblumSolomyak} and references therein).
\end{enumerate}
\end{enumerate}
In this paper we address the item \ref{iBohr} above in detail.
Items \ref{iBa}, \ref{iBb}, \ref{iBc} will be the subject of future work. 
\begin{figure}[htb]
\centering
\includegraphics[width=0.5\textwidth]{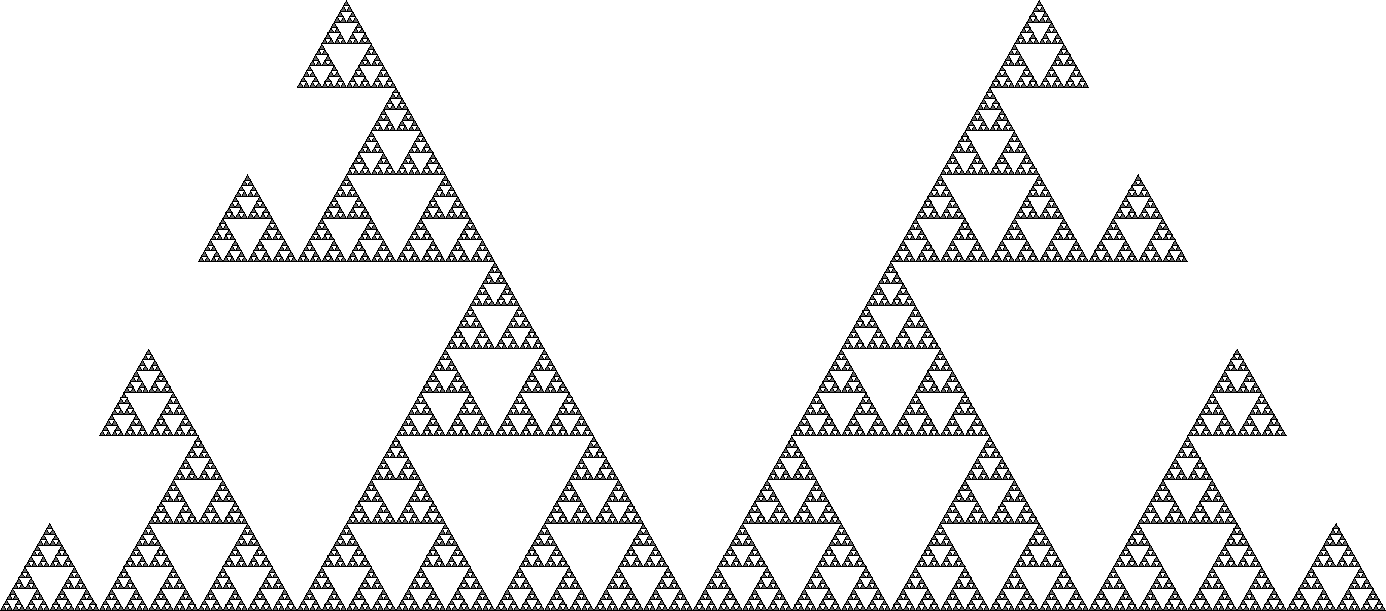}
\caption{Part of an infinite blow-up of $SG(2)$, which is Type (i) of the fractafold considered in \S\ref{sec:fractafold}.}
\label{fig:InfiniteBlowup}
\end{figure}
\begin{figure}[htb]
\centering
\includegraphics[width=0.5\textwidth]{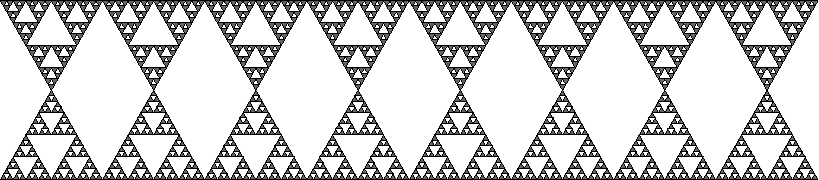}
\caption{The ladder periodic fractafold based on $SG(2)$, which is Type (ii) of the fractafold considered in \S\ref{sec:fractafold}.}
\label{fig:LadderFractafold}
\end{figure}
\begin{figure}[htb]
\centering
\includegraphics[width=0.5\textwidth]{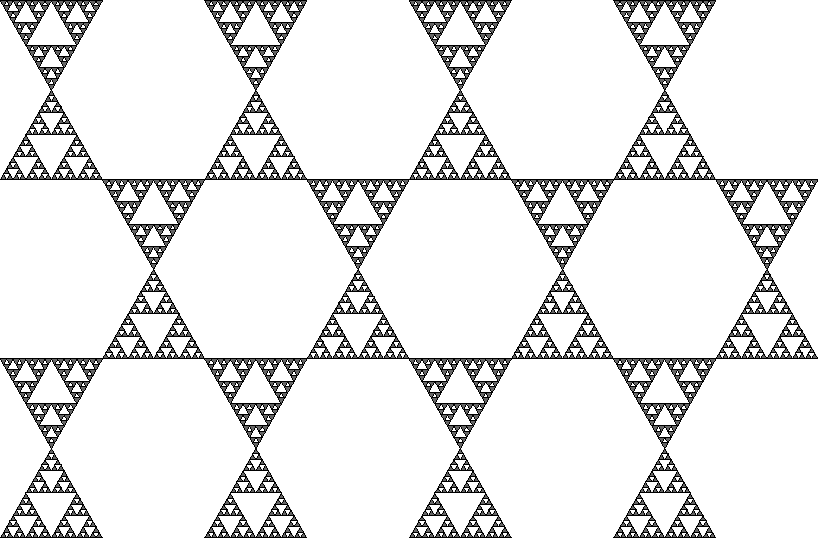}
\caption{The hexagonal periodic fractafold based on $SG(2)$, which is Type (ii) of the fractafold considered in \S\ref{sec:fractafold}.}
\label{fig:HexagonalFractafold}
\end{figure}
\begin{figure}[htb]
\centering
\includegraphics[width=0.5\textwidth]{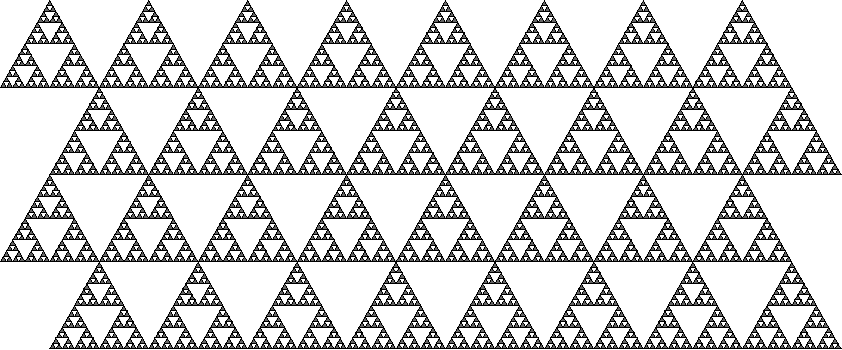}
\caption{The triangular lattice finitely ramified fractal field based on $SG(2)$, which is considered in \S\ref{fr-fi}.}
\label{fig:TriFractalfield}
\end{figure}

Our main objective is to consider, instead of the Euclidean space, a
\emph{fractafold}, which according to Strichartz \cite{StrFractafold} is defined as ``a space that is locally modeled on a specified fractal, the fractal equivalent of a manifold.'' The first instance of a fractafold is  
 the infinite 
 Sierpinski gasket (Figure 
 \ref{fig:InfiniteBlowup}).
  As shown by Barlow and Perkins \cite{BP}, the heat kernel $p_t(x,y)$ on the infinite Sierpinski gasket satisfies a sub-Gaussian estimate
  with respect to the Euclidean metric $d(\cdot,\cdot)$:
 $$ p_t(x,y)\asymp c_1 t^{-d_s/2}\exp\left( - c_2 \left(\frac{d(x,y)^{d_w}}{t}\right)^{1/(d_w-1)} \right),$$
 where $d_s = 2\log3/\log 5$ and $d_w = \log 5/\log2 > 2$. Here $\asymp$ means that there are upper and lower estimates, 
 but the constants $c_1,c_2$ in them may be different.
 We would like to note, however, that the heat kernel 
 is not immediately relevant for spectral analysis, partially because its form is complicated, but mostly because the domain of the Laplacian is not an algebra under multiplication \cite{BST}. 
 Other typical examples    of fractafolds that we consider, see \cite{StrTep} and Section \ref{sec:examples} for details, are shown in Figures \ref{fig:LadderFractafold}, 
 \ref{fig:HexagonalFractafold}, and \ref{fig:TriFractalfield}. 
 For background concerning spectral analysis on fractafolds, see  \cites{s1,s2,eigen,i1,i2,i3,RT,RST,StrichartzSHO,o1,OSt,OS-CT,r1,StrichartzFractalsInTheLarge,T,q}. 
 Existence of gaps in the spectrum is investigated in 
 \cites{g1,g2,g3,Hare}.
 Wave equation on fractals is discussed  in 
 \cites{s0,w1,w2,w3,w4}.
 Physics applications, and spectral zeta functions in particular, are given in  
 \cites{ph1,ph-b,A,Dunne,Ben,Tzeta,dgv}.

\section{Main results} \label{sec:mainresults}
 
\subsection{Spectral asymptotics of $-\Delta+V$} \label{sec:setup}

In all the examples to follow, $K$ is a compact set in $\mathbb{R}^d$ endowed with a Borel probability measure $\mu$ and a ``well-defined boundary'' $\partial K$ which has $\mu$-measure zero. We shall assume that there exists a well-defined self-adjoint Laplacian operator $-\Delta^\wedge$ (resp. $-\Delta^\vee$) on $L^2(K,\mu)$ satisfying the Dirichlet (resp. Neumann) condition on $\partial K$. Note that $\partial K$ might not coincide with the boundary of $K$ in the topological sense. We assume (as is well known in examples) that both $-\Delta^\wedge$ and $-\Delta^\vee$ have compact resolvents, and hence have pure point spectra. It then makes sense to introduce the eigenvalue counting function
\begin{equation}
\label{eq:ECF}
N^{\rm b}(K,\mu, \lambda) := \#\left\{ \lambda_i(-\Delta^{\rm b})\leq \lambda\right\},\quad {\rm b}\in\{\wedge,\vee\}.
\end{equation}

\begin{assumption}
\label{ass:existspecdim}
There exists a positive constant $d_s$ such that
\begin{equation}
0<\varliminf_{\lambda\to\infty} \lambda^{-d_s/2}  N^{\rm b}(K,\mu,\lambda) \leq \varlimsup_{\lambda\to\infty} \lambda^{-d_s/2} N^{\rm b}(K,\mu,\lambda) <\infty,
\end{equation}
where ${\rm b} \in \{\wedge,\vee\}$.
\end{assumption}

A stronger condition than Assumption \ref{ass:existspecdim} is

\begin{assumption}[Weyl asymptotics of the bare Laplacian]
\label{ass:Weyl}
There exist a positive constant $d_s$ and a right-continuous 
with left limits (c\`{a}dl\`{a}g), $T$-periodic function $G:\mathbb{R} \to \mathbb{R}_+$ satisfying
\begin{enumerate}[label=(G\arabic*)]
\item $0<\inf G \leq \sup G <\infty$.
\item $G$ is independent of the boundary condition ${\rm b}\in \{\wedge,\vee\}$
\end{enumerate}
 such that as $\lambda\to\infty$,
\begin{equation}
\label{eq:Weyl}
N^{\rm b}(K,\mu,\lambda) = \lambda^{d_s/2} \left[G\left(\frac{1}{2}\log\lambda\right)+ R^{\rm b}(\lambda)\right],
\end{equation}
where $R^{\rm b}(\lambda)$ denotes the remainder term of order $o(1)$.
\end{assumption}

\begin{remark}
The parameter $d_s$ is often identified with the \emph{spectral dimension} of the bare Laplacian $-\Delta$ on $L^2(K,\mu)$. If $K$ is a domain in $\mathbb{R}^d$ with a nice boundary, and $\mu$ is the Lebesgue measure, then $d_s =d$ and $G$ is an explicit constant $(2\pi)^{-d}\mu(B)\mu(K)$, where $B$ is the unit ball in $\mathbb{R}^d$. However, there are classes of fractals $K$ for which (\ref{eq:Weyl}) holds with $G$ being possibly nonconstant.

In many examples, the leading-order term in $R^{\rm b}(\lambda)$ gives information about the boundary of the domain. For an Euclidean domain in $\mathbb{R}^d$ with nice boundary, the leading-order term of $R^{\rm b}(\lambda)$ scales with $\lambda^{-1/2}$, and the sign of this term is negative (resp. positive) if ${\rm b}=\wedge$ (resp. if ${\rm b}=\vee$) \cites{BrCa,FlVa,Iv,LaPo}. 

\end{remark}

We now consider an unbounded space $K_\infty$ which admits a cellular decomposition into copies of $K$. Formally, let $K_\infty := \cup_\alpha K_\alpha$, where 
\begin{itemize}
\item Each $K_\alpha$ is isometric to $K$ via the map $\phi_\alpha: K\to K_\alpha$.
\item We identify $\partial K_\alpha :=\phi_\alpha(\partial K)$ to be the boundary of $K_\alpha$, and $K_\alpha^\circ:=K_\alpha \backslash \partial K_\alpha$ the interior of $K_\alpha$.
\item (Cells adjoin only on the boundary.) For all $\alpha \neq \alpha'$, $\left(K_\alpha \cap K_{\alpha'}\right) =\left(\partial K_\alpha \cap \partial K_{\alpha'}\right)$.
\end{itemize}
Let $\mu_\alpha:= \mu \circ \phi_\alpha^{-1}$ be the push-forward measure of $\mu$ onto $K_\alpha$. For any $\alpha \neq \alpha'$, it is direct to define the ``glued'' measure $\mu_{\alpha,\alpha'}$ on $K_\alpha \cup K_{\alpha'}$ in the natural way:
\begin{equation}
\forall B \in \mathcal{B}(K_\alpha \cup K_{\alpha'}) : \qquad \mu_{\alpha,\alpha'}(B) = \mu_\alpha(B\cap K_\alpha) + \mu_{\alpha'}(B \cap K_{\alpha'}).
\end{equation}
By extension we define the measure $\mu_\infty$ on $K_\infty$.

\begin{proposition}[Decoupling of 
$L^2$]
\label{prop:decoupling}
For all $\alpha \neq \alpha'$ we have $K^\circ_\alpha \cap K^\circ_{\alpha'} =\emptyset$ and $L^2(K^\circ_\alpha \cup K^\circ_{\alpha'}, \mu_{\alpha,\alpha'}) = L^2(K^\circ_\alpha, \mu_\alpha) \oplus L^2(K^\circ_{\alpha'}, \mu_{\alpha'})$.
\end{proposition}

Proposition \ref{prop:decoupling} allows one to decouple the Laplacian on the glued measure space into the direct sum of the Laplacians on the individual components (see \cite{ReedSimonVol4}*{Proposition XIII.15.3}):
\begin{equation}
\Delta^{\rm b}_{K_\alpha\cup K_{\alpha'}}: = \Delta^{\rm b}_{K_\alpha} \oplus \Delta^{\rm b}_{K_{\alpha'}},
\end{equation}
from which it follows that
\begin{equation}
N^{\rm b}(K_\alpha\cup K_{\alpha'}, \mu_{\alpha,\alpha'},\lambda) = N^{\rm b}(K_\alpha, \mu_\alpha,\lambda) + N^{\rm b}(K_{\alpha'}, \mu_{\alpha'}, \lambda).
\end{equation}
By extension we have that
\begin{equation}
N^{\rm b}(K_\infty,\mu_\infty,\lambda) = \sum_\alpha N^{\rm b}(K_\alpha,\mu_\alpha,\lambda).
\end{equation}

For future purposes we also put a metric $d: K_\infty\times K_\infty \to [0,\infty)$, and fix an origin $0\in K_\infty$. In proving our main results, the metric $d$ does not play a major role. However for practical applications, such as determining the spectral dimension of the Schr\"odinger operator, one needs to understand the interplay between the metric $d$ and the measure $\mu_\infty$; see Remark \ref{rem:pot} and Section \ref{sec:examples}. 

Let the potential $V$ be a nonnegative, locally bounded measurable function on $K_\infty$. (In general, $V$ can be a real-valued, locally bounded measurable function which is bounded below. By adding a suitable constant to $V$ one retrieves the case of a nonnegative potential.)
  
\begin{assumption}
\label{ass:sa}
There exists a self-adjoint Laplacian $-\Delta$ on $L^2(K_\infty,\mu_\infty)$ [equivalently, a local regular Dirichlet form $(\widetilde{\mathcal{E}},\widetilde{\mathcal{F}})$ on $L^2(K_\infty, \mu_\infty)$], and that the potential $V(x) \to +\infty$ as $d(0,x)\to+\infty$.
\end{assumption}  

\begin{proposition}
\label{prop:pp}
Under Assumption \ref{ass:sa}, the Schr\"odinger operator $(-\Delta +V)$, regarded as a sum of quadratic forms, is self-adjoint on $L^2(K_\infty,\mu_\infty)$, and has pure point spectrum.
\end{proposition}
\begin{proof}
This uses the min-max principle as stated in \cite{ReedSimonVol4}*{Theorem XIII.2}, and then follows the proof of \cite{ReedSimonVol4}*{Theorem XIII.16}.
\end{proof}

By virtue of Proposition \ref{prop:pp}, we can define the eigenvalue counting function for $(-\Delta+V)$ on $K_\infty$:
\begin{equation}
N(K_\infty, \mu_\infty,V,\lambda) := \#\left\{\lambda_i\left(-\Delta+V\right)\leq \lambda\right\}.
\end{equation}
We are interested in the asymptotics of $N(K_\infty,\mu_\infty,V,\lambda)$ as $\lambda\to\infty$. In order to state the precise results, we will impose some mild conditions on the potential $V$. 

Given a potential $V$ on $K_\infty$, let $V^\wedge$ (resp. $V^\vee$) be the function which is piecewise constant on each cell $K_\alpha$, and takes value $\sup_{x\in K_\alpha} V(x)$ (resp. $\inf_{x\in K_\alpha} V(x)$) on $K_\alpha$. We introduce the associated distribution functions
\begin{eqnarray}
F^\wedge(V,\lambda) &:=& \mu_\infty\left( \{x\in K_\infty: V^\wedge(x) \leq \lambda\}\right),\\
F^\vee(V,\lambda) &:=& \mu_\infty\left( \{x\in K_\infty: V^\vee(x) \leq \lambda\}\right).
\end{eqnarray}
Note that $F^\wedge(V,\lambda) \leq F^\vee(V,\lambda)$.

\begin{assumption}
\label{ass:V-s}
There exists a constant $C>0$ such that $F^\vee(V, 2\lambda) \leq C F^\wedge(V,\lambda)$ for all sufficiently large $\lambda$.
\end{assumption}

Note that this assumption implies that both $F^\vee(V,\cdot)$ and $F^\wedge(V,\cdot)$ have the doubling property: there exist $C^\vee, C^\wedge>0$ such that
\begin{equation}
F^\vee(V,2\lambda) \leq C^\vee F^\vee(V,\lambda) \quad\text{and}\quad F^\wedge(V,2\lambda) \leq C^\wedge F^\wedge(V,\lambda)
\end{equation}
for all sufficiently large $\lambda$.

\begin{assumption}
\label{ass:V}
The potential $V$ on $K_\infty$ satisfies
\begin{equation}
\frac{F^\vee(V,\lambda)}{F^\wedge(V,\lambda)} = 1+o(1) \quad \text{as}~\lambda\to\infty.
\end{equation}
\end{assumption}

\begin{remark}\label{rem:pot}
To understand Assumption \ref{ass:V-s} or \ref{ass:V}, it helps to keep the following example in mind. Let $(K_\infty, \mu_\infty, d)$ be a metric measure space which admits a cellular decomposition into copies of the compact metric measure space $(K,\mu,d)$. Let ${\rm diam}_d(K)$ be the diamater of $K$ in the $d$-metric. Further suppose that $\mu_\infty$ is Ahlfors-regular: there exist positive constants $c_1$, $c_2$, and $\alpha$ such that 
\begin{equation}
c_1 r^\alpha \leq \mu_\infty(B_d(x,r)) \leq c_2 r^\alpha
\end{equation}
for all $x\in K_\infty$ and sufficiently large $r>0$. As for the potential $V$, assume that there exist $\beta>1$ and $\gamma \in (0,1]$ such that
\begin{equation}\label{eq:Vdb}
c_3 d(0,x)^\beta \leq V(x) \leq c_4 d(0,x)^\beta,
\end{equation}
\begin{equation}
\label{eq:Vreg}
\frac{ |V(x)-V(y)|}{d(x,y)^\gamma} \leq c_5 [\max(d(0,x), d(0,y))]^{\beta-\gamma}
\end{equation}
for all $x,y \in K_\infty$. By a direct calculation one can verify that (\ref{eq:Vdb}) implies
\begin{equation}
c_6 \lambda^{\alpha/\beta}\leq F^{\rm b}(V,\lambda) \leq c_7 \lambda^{\alpha/\beta},
\end{equation}
which satisfies Assumption \ref{ass:V-s}. Meanwhile, (\ref{eq:Vreg}) implies
\begin{equation}
V^\wedge(x)-V^\vee(x) \leq c_8 [{\rm diam}_d(K)]^\gamma d(0,x)^{\beta-\gamma} .
\end{equation}
Thus (\ref{eq:Vdb}) and (\ref{eq:Vreg}) together imply Assumption \ref{ass:V}.  
\end{remark}

Our main results are the following.

\begin{theorem}[Existence of spectral dimension]
\label{thm:specdim}
Under Assumptions \ref{ass:existspecdim}, \ref{ass:sa}, and \ref{ass:V-s}, we have that
\begin{equation}
0<\varliminf_{\lambda\to\infty} \frac{N(K_\infty, \mu_\infty, V,\lambda)}{\lambda^{d_s/2} F(V,\lambda)} \leq \varlimsup_{\lambda\to\infty} \frac{N(K_\infty,\mu_\infty,V,\lambda)}{\lambda^{d_s/2}F(V,\lambda)} <\infty,
\end{equation}
where $F(V,\lambda) := \mu_\infty\left(\{x\in K_\infty: V(x)\leq \lambda\} \right)$. In particular, if $F(V,\lambda) = \Theta(\lambda^\beta)$ as $\lambda\to\infty$, then $d_s(V) = d_s + 2\beta$ is the effective spectral dimension of the Schr\"odinger operator $(-\Delta +V)$.
\end{theorem}

\begin{theorem}[Bohr's formula]
\label{thm:Bohrmain}
Under Assumptions \ref{ass:Weyl}, \ref{ass:sa}, and \ref{ass:V},
\begin{equation}
\lim_{\lambda\to\infty} \frac{N(K_\infty,\mu_\infty,V,\lambda)}{g(V,\lambda)}=1,
\end{equation}
where
\begin{equation}
g(V,\lambda) :=  \int\limits_{K_\infty} \left[\left(\lambda-V(x)\right)_+ \right]^{d_s/2}G\left(\frac{1}{2}\log(\lambda-V(x))_+\right)\, \mu_\infty(dx),
\end{equation}
and $(f)_+ = \max\{f,0\}$.
\end{theorem}

In what follows we shall refer to $g$ as ``Bohr's asymptotic function.'' 

The proof of Theorem \ref{thm:Bohrmain}, discussed in Section \ref{sec:Bohr}, utilizes Dirichlet-Neumann bracketing on the eigenvalue counting function and on Bohr's asymptotic function. This is a relatively standard technique which is explained in the mathematical physics literature; see \emph{e.g.} \cite{ReedSimonVol4}*{\S XIII}. The novelty of our approach is to restate the sufficient condition on the potential $V$ in terms of its distribution function, which allows us to extend the classical Bohr's formula to a wider class of settings, such as on unbounded fractal spaces. 

\subsection{Laplace transform version}

There are also analogs of Theorems \ref{thm:specdim} and \ref{thm:Bohrmain} for the Laplace-Stieltjes transform of the eigenvalue counting function
\begin{equation}
\mathcal{L}(K_\infty,\mu_\infty,V,t) := {\rm Tr}_{K_\infty}\{e^{-t(-\Delta+V)}\} = \int_0^\infty e^{-\lambda t}\,N(K_\infty,\mu_\infty,V,d\lambda) .
\end{equation}
When $V=0$ this is the trace of the heat semigroup associated with the bare Laplacian $-\Delta$. More generally, it can be regarded as the trace of the Feynman-Kac semigroup associated to the Markov process driven by $-\Delta$ subject to kiling with rate $V(x)$ at $x\in K_\infty$.

The reason for stating the analog versions is because for certain compact metric measure spaces, it is not known whether an explicit Weyl asymptotic formula for the bare Laplacian (Assumption \ref{ass:Weyl}) exists. However it may be the case that an asymptotic formula for the \emph{heat kernel trace} (in some literature it is also called the \emph{partition function})
\begin{equation}
\mathcal{L}(K,\mu,t) := {\rm Tr}\{e^{t\Delta}\} = \int\limits_K p_t(x,x) \,\mu(dx)
\end{equation}
exists in the $t\downarrow 0$ limit. Here $p_t(x,y) ~(t>0,~x,y\in K)$ is the heat kernel associated to the Markov semigroup $e^{t\Delta}$ generated by the self-adjoint Laplacian $-\Delta$ on $L^2(K,\mu)$. To be more precise, we denote by $\mathcal{L}^{\rm b}(K,\mu,t)$ the heat kernel trace of the Laplacian $-\Delta^{\rm b}$ on $L^2(K,\mu)$ with boundary condition ${\rm b} \in \{\wedge,\vee\}$. Then
\begin{equation}
\mathcal{L}^{\rm b}(K,\mu, t) = \int_0^\infty e^{-\lambda t} N^{\rm b}(K,\mu,d\lambda) = \int_K \, p_t^{\rm b}(x,x)\,\mu(dx),
\end{equation}
where $N^{\rm b}(K,\mu,\lambda)$ is as in (\ref{eq:ECF}), and $p_t^{\rm b}(x,y)$ is the heat kernel associated with the infinitesimal generator $-\Delta^{\rm b}$. 

%

\begin{assumption}[Existence of the spectral dimension for the bare Laplacian]
\label{ass:SpecHKT}
There exists a positive constant $d_s$ such that
\begin{equation}
0< \varliminf_{t\downarrow 0} t^{d_s/2}\mathcal{L}^{\rm b}(K,\mu,t) \leq \varlimsup_{t\downarrow 0} t^{d_s/2}\mathcal{L}^{\rm b}(K,\mu,t) <\infty
\end{equation}
for ${\rm b} \in \{\wedge,\vee\}$.
\end{assumption}


\begin{theorem}
\label{thm:specdimHKT}
Under Assumptions \ref{ass:sa}, \ref{ass:V-s}, and \ref{ass:SpecHKT}, we have that
\begin{equation}
0 < \varliminf_{t\downarrow 0} \frac{\mathcal{L}(K_\infty,\mu_\infty,V,t)}{t^{-d_s/2} \mathcal{F}(V,t)} \leq \varlimsup_{t\downarrow 0} \frac{\mathcal{L}(K_\infty,\mu_\infty,V,t)}{t^{-d_s/2} \mathcal{F}(V,t)} <\infty,
 \end{equation} 
where
\begin{equation}
\mathcal{F}(V,t) = \int_{K_\infty} \, e^{-t V(x)}\,\mu_\infty(dx).
\end{equation}
In particular, if $F(V,\lambda) := \mu_\infty\left(\{x\in K_\infty: V(x)\leq \lambda\} \right)=\Theta(\lambda^\beta)$ as $\lambda\to\infty$, then $d_s(V)=d_s+2\beta$ is the spectral dimension for the Schr\"odinger operator $(-\Delta+V)$.
\end{theorem}
\begin{assumption}[Weak Weyl asymptotics for the bare Laplacian]
\label{ass:WeylHKT}
There exists a positive constant $d_s$ and a 
continuous function $H:\mathbb{R}_+\to\mathbb{R}_+$, independent of the boundary condition ${\rm b}\in \{\wedge,\vee\}$ and with $0 < \inf H \leq \sup H < \infty$, such that as $t\downarrow 0$,
\begin{equation}
\mathcal{L}^{\rm b}(K,\mu, t) = t^{-d_s/2}\left[H(t) + \rho^{\rm b}(t)\right],
\end{equation}
where $\rho^{\rm b}(t)$ denotes the remainder term of order $o(1)$.

\end{assumption}
\begin{theorem}[Laplace transform version of Bohr's formula]
\label{thm:LaplaceBohr}
Under Assumptions \ref{ass:sa}, \ref{ass:WeylHKT}, and \ref{ass:V}, we have that
\begin{equation}
\label{eq:LaplaceBohr}
\lim_{t\downarrow 0} \frac{\mathcal{L}(K_\infty,\mu_\infty,V,t)}{t^{-d_s/2}H(t)\mathcal{F}(V,t)}=1,
\end{equation}
\end{theorem}

Note that (\ref{eq:LaplaceBohr}) can also be interpreted as the asymptotic factorization of the trace of the Feynman-Kac semigroup:
\begin{equation}
\lim_{t\downarrow 0} \frac{{\rm Tr}_{K_\infty}\{e^{-t(-\Delta+V)}\}}{ {\rm Tr}_K\{e^{t\Delta}\} \cdot {\rm Tr}_{K_\infty}\{e^{-tV}\}}=1.
\end{equation}

\begin{remark}
We make a few comments concerning the connections between Assumption \ref{ass:existspecdim}/\ref{ass:Weyl} and Assumption \ref{ass:SpecHKT}/\ref{ass:WeylHKT}.
\begin{enumerate}[label=(\roman*)]
\item Assumption \ref{ass:existspecdim} is equivalent to Assumption \ref{ass:SpecHKT}. 
\item Assumption \ref{ass:Weyl} implies Assumption \ref{ass:WeylHKT}.
However, the reverse implication is possibly not true, since the classical technique of Tauberian theorems 
may not be applicable
in this context.
\item In order to prove Bohr's formula (Theorem \ref{thm:Bohrmain}), we impose in Assumption \ref{ass:Weyl} that the function $G$ be a periodic function. This is natural in light of the fractal examples we are interested in. However, to prove the Laplace transform version of Bohr's formula (
Theorem \ref{thm:LaplaceBohr}
), one does not need to assume the log-periodicity in Assumption \ref{ass:WeylHKT}. This leads to the question of whether one could relax the periodicity of $G$ and still be able to prove the original Bohr's formula in greater generality (we do not address this question in the present work).
\end{enumerate}
\end{remark}

\subsection{Application of the main results}

To illustrate how our main results can be used, we now describe the ``harmonic oscillator'' problem on the Sierpinski gasket which was investigated in \cite{StrichartzSHO}. For discussions of more general unbounded potentials on other fractal-like spaces, see Section \ref{sec:examples}.

\begin{example}[Harmonic oscillator on the infinite blow-up of the Sierpinski gasket]
Let $K$ be the Sierpinski gasket ($SG$). To construct $SG$, we first set the three vertices $\{p_1, p_2, p_3\}$ of an equilateral triangle in $\mathbb{R}^2$, and then introduce the contraction maps $\Psi_j: \mathbb{R}^2 \to \mathbb{R}^2$, $\Psi_j(x) = \frac{1}{2}(x-p_j) + p_j$, $j=1,2,3$. Then $SG$ is the unique fixed point $K$ under the iterated function system consisting of the $\Psi_j$: $K = \cup_{j=1}^3 \Psi_j(K)$. Let $w=w_1 w_2 \cdots w_m$ be a word of length $|w|=m$ where each letter $w_j \in \{1,2,3\}$, and define the map $\Psi_w = \Psi_{w_1} \circ \cdots \circ \Psi_{w_m}$.

We endow $SG$ with the uniform self-similar measure $\nu$ with $\nu(\Psi_w K) = 3^{-|w|}$. The theory of Kigami \cite{Kigami} allows us to define the standard Laplacian on $L^2(SG, \nu)$ with either Dirichlet or Neumann condition on the boundary $\partial(SG)=\{p_1,p_2,p_3\}$. Moreover, Kigami and Lapidus \cite{KigamiLapidus} proved that the eigenvalue counting function for the standard Laplacian satisfies
\begin{equation}
N^{\rm b}(SG,\nu,\lambda) = \lambda^{d_s/2} \left[G\left(\frac{1}{2}\log\lambda\right)+o(1)\right] \qquad ({\rm b} \in \{\wedge,\vee\}),
\end{equation}
where $d_s= 2\log3/\log 5$, and $G$ is a c\`{a}dl\`{a}g periodic function with period $\frac{1}{2}\log 5$ and contains discontinuities. Thus Assumption \ref{ass:Weyl} is satisfied.

Next, for each infinite word $w=w_1 w_2 \cdots$ which is not eventually constant, define
\begin{equation}
SG_\infty^w := \bigcup_{m=0}^\infty \left(\Psi_{w_1}^{-1}\circ \cdots \circ \Psi_{w_m}^{-1}\right)(SG)
 \end{equation}
to be the \emph{infinite blow-up} of $SG$ associated with the word $w$. This is an unbounded fractal space where the neighborhood of any point $x\in K_\infty$ is homeomorphic to $SG$, and thus is a fractal analog of a manifold, called a \emph{fractafold} by Strichartz \cite{StrichartzFractalsInTheLarge}. Properties of the Laplacian on $SG_\infty^w$ are discussed in \cite{StrFractafold}. Here we point out that by construction, $SG_\infty^w$ admits a cellular decomposition into copies of $SG$ which intersect on the boundary only. Thus the measure $\nu$ on $SG$ can be readily extended to the measure $\nu_\infty$ on $SG_\infty^w$.

In \cite{StrichartzSHO} Fan, Khandker, and Strichartz studied the spectral problem of a harmonic oscillator potential $V$ on a class of infinite blow-ups of $SG$. They defined $V$ to be a solution to $-\Delta V=-1$ on $SG_\infty^w$ which grows unboundedly as $d(0,x)\to\infty$ and attains a minimum at some vertex $x_0\in K_\infty$. (The first condition is a suitable replacement of $V(x)=\frac{1}{2}|x|^2$, which is available only in the Euclidean setting.)
Note that this implies that $V(x)$ grows at infinity at rate 
comparable to a positive power of $ R(x_0,x)$, where $R(\cdot,\cdot)$ is the effective resistance metric on $SG_\infty^w$. This verifies Assumption \ref{ass:V-s}. However we cannot verify Assumption \ref{ass:V} for general words $w$. 
Paper \cite{StrichartzSHO} also contains information about spectral dimension, which 
depends on the blow-ups of $SG$. 
Through a mix of computations and numerical simulations, the authors of 
\cite{StrichartzSHO}
 were able to find properties of the low-lying eigenfunctions, as well as the asymptotic growth rate of the eigenvalue counting function of $-\Delta+V$ \cite{StrichartzSHO}*{Theorem 8-1 and Eq. (8.18)}:
\begin{equation}
\label{SGHarmscaling}
c\lambda^{d_s} \leq N(SG_\infty^w, \nu_\infty, V,\lambda) \leq C \lambda^{d_s}.
\end{equation}
Among the open questions posed in \cite{StrichartzSHO}*{Problem 8-3 and Conjecture 8-4} is finding the asymptotic ``Weyl ratio'' $\lambda^{-d_s(V)/2}N(K_\infty,\mu_\infty,V,\lambda)$ of the eigenvalue counting function. Here we can provide an indirect answer. Given that Assumptions \ref{ass:Weyl}, \ref{ass:sa}, and \ref{ass:V} are satisfied, Bohr's formula (Theorem \ref{thm:Bohrmain}) says that as $\lambda\to\infty$,
\begin{equation}\label{eq:SGBohr}
N(SG_\infty^w,\nu_\infty,V,\lambda) =(1+o(1)) \int_{SG_\infty^w} \, \left[\left(\lambda-V(x)\right)_+\right]^{d_s/2} G\left(\frac{1}{2}\log(\lambda-V(x))_+\right)\,d\nu_\infty(x).
\end{equation}
This in some sense answers the Weyl ratio question, in spite of the non-explicit nature of the integral on the right-hand side.
\end{example}

The rest of this paper is organized as follows. In Section \ref{sec:Bohr} we describe the tools needed to establish Bohr's formula in the setting of an unbounded space which admits a cellular decomposition according to 
the setup in Section \ref{sec:setup}. In Section \ref{sec:Bohrpot} we show how to restate the general sufficient condition for Bohr's formula in terms of distribution functions of $V^\vee$ and $V^\wedge$, and also give a ``weak'' version of Bohr's formula. 
We can show how the addition of an unbounded potential leads to the absence of gaps in the spectrum of $-\Delta +V$. This is of independent interest since the spectrum of the bare Laplacian on certain fractals (\emph{e.g.} the Sierpinski gasket) has gaps. In Section \ref{sec:LaplaceBohr} we establish the Laplace transform version of Bohr's formula. Finally, in Section \ref{sec:examples}, we discuss applications of our main results to various unbounded potentials on several types of unbounded fractal spaces.

\section{The general Bohr's formula}\label{sec:Bohr}

In this section and the next section, Assumptions \ref{ass:Weyl} and \ref{ass:sa} are in force.

\subsection{Bohr's asymptotic functions}
  
Let $-\Delta^\wedge$ (resp. $-\Delta^\vee$) be the Laplacian on $L^2(K_\infty,\mu_\infty)$ with Dirichlet (resp. Neumann) conditions on the gluing boundary $\cup_\alpha \partial K_\alpha$.  For each potential $V$, let $V^\wedge$ (resp. $V^\vee$) be the piecewise constant function which takes value $\sup_{x\in K_\alpha} V(x)$ (resp. $\inf_{x\in K_\alpha} V(x)$) on $K_\alpha$. Thanks to Proposition \ref{prop:pp}, one can introduce the eigenvalue counting functions
\begin{eqnarray}
N(K_\infty, \mu_\infty,V,\lambda) &:=& \#\left\{\lambda_i \left(-\Delta+V\right)\leq \lambda\right\},\\
N^\wedge(K_\infty, \mu_\infty,V,\lambda) &:=& \#\left\{\lambda_i \left(-\Delta^\wedge+V^\wedge\right)\leq \lambda\right\},\\
N^\vee(K_\infty, \mu_\infty,V,\lambda) &:=& \#\left\{\lambda_i \left(-\Delta^\vee+V^\vee\right)\leq \lambda\right\}.
\end{eqnarray}
Note that since $(-\Delta^\vee + V^\vee) \leq (-\Delta +V) \leq (-\Delta^\wedge+ V^\wedge)$ in the sense of quadratic forms,
\begin{equation}
\label{eq:DNbracket}
N^\wedge(K_\infty, \mu_\infty,V,\lambda)\leq N(K_\infty, \mu_\infty,V,\lambda) \leq N^\vee(K_\infty, \mu_\infty,V,\lambda).
\end{equation}
 
We shall show that under some mild additional conditions on $V$, $N(K_\infty, \mu_\infty, V,\lambda)$ is asymptotically comparable to the ``Bohr's asymptotic function''
\begin{equation}
\label{eq:g}
g(V,\lambda) := \int\limits_{K_\infty} \left[\left(\lambda-V(x)\right)_+ \right]^{d_s/2}G\left(\frac{1}{2}\log(\lambda-V(x))_+\right)\, d\mu_\infty(x),
\end{equation}
where $(f)_+:= \max\{f,0\}$, and $G$ is as appeared in Assumption \ref{ass:Weyl}. In order to estimate this rate of convergence, we introduce the functions
\begin{equation}
\label{eq:gb}
g^{\rm b}(V,\lambda) := \int\limits_{K_\infty} \left[\left(\lambda-V^{\rm b}(x)\right)_+ \right]^{d_s/2}G\left(\frac{1}{2}\log(\lambda-V^{\rm b}(x))_+\right)\, d\mu_\infty(x)
\end{equation}
and
\begin{equation}
\label{eq:Rb}
\mathcal{R}^{\rm b}(V,\lambda) := \int\limits_{K_\infty} \left[(\lambda-V^{\rm b}(x))_+\right]^{d_s/2} R^{\rm b}\left((\lambda-V^{\rm b}(x))_+\right)\,d\mu_\infty(x)
\end{equation}
for ${\rm b}\in\{\wedge, \vee\}$, where $R^{\rm b}$ is the remainder term which appeared in Assumption \ref{ass:Weyl}. Observe that since $V^{\rm b}(x)$ is constant on cells, the right-hand side expressions in (\ref{eq:gb}) and (\ref{eq:Rb}) are really discrete sums:
\begin{eqnarray}
\label{eq:gb2} g^{\rm b}(V,\lambda) &=&\sum_{\{\alpha: \left.V^{\rm b}\right|_{K_\alpha}\leq \lambda\}}\left[\lambda-\left.V^{\rm b}\right|_{K_\alpha}\right]^{d_s/2} G\left(\frac{1}{2}\log\left(\lambda-\left.V^{\rm b}\right|_{K_\alpha}\right)\right),\\
\label{eq:Rb2} \mathcal{R}^{\rm b}(V,\lambda) &=&\sum_{\{\alpha: \left.V^{\rm b}\right|_{K_\alpha}\leq \lambda\}}\left[\lambda-\left.V^{\rm b}\right|_{K_\alpha}\right]^{d_s/2} R^{\rm b}\left(\lambda-\left.V^{\rm b}\right|_{K_\alpha}\right).
\end{eqnarray}
Moreover, by Proposition \ref{prop:decoupling}, $K_\infty$ decouples into the various $K_\alpha$ according to the Dirichlet or Neumann boundary condition, so
\begin{equation}
\label{eq:Nbreakdown}
N^{\rm b}(K_\infty,\mu_\infty,V,\lambda) = \sum_{\{\alpha: \left.V^{\rm b}\right|_{K_\alpha} \leq \lambda\}} N^{\rm b}\left(K_\alpha,\,\mu_\alpha,\,\lambda-\left.V^{\rm b}\right|_{K_\alpha}\right).
\end{equation}
Pulling (\ref{eq:Weyl}), (\ref{eq:gb2}), (\ref{eq:Rb2}), and (\ref{eq:Nbreakdown}) together we obtain
\begin{equation}
\label{eq:NgR}
N^{\rm b}(K_\infty, \mu_\infty, V,\lambda)=g^{\rm b}(V,\lambda) + \mathcal{R}^b(V,\lambda) .
\end{equation}

\subsection{Monotonicity of Bohr's asymptotic functions}

A key monotonicity result we need is

\begin{proposition}
\label{prop:monotoneg}
Fix a potential $V$. Then each of the functions $\lambda \mapsto g(V,\lambda)$, $\lambda\mapsto g^\wedge(V,\lambda)$, and $\lambda\mapsto g^\vee(V,\lambda)$ is monotone nondecreasing for all $\lambda>0$. Moreover, $g^\wedge(V,\lambda) \leq g(V,\lambda) \leq g^\vee(V,\lambda)$.
\end{proposition}

This follows from the monotonicity of the integrand of the $g$ function.

\begin{proposition}
\label{prop:monotone}
The function $W(\lambda)=\lambda^{d_s/2} G\left(\frac{1}{2}\log\lambda\right)$ is 
nondecreasing.
\end{proposition}

\begin{remark}
This proposition implies, in particular, that $G$ has a c\`{a}dl\`{a}g version. 
Although the result is very simple, we could not find it in the literature.
\end{remark}

\begin{proof}[Proof of Proposition \ref{prop:monotone}]
If  $W$ was not  nondecreasing, then there would existed   $\lambda_2>\lambda_1>0$ such that $W(\lambda_2) =W(\lambda_1)=-\delta<0$, which contradicts 
\eqref{eq:Weyl}
and the monotonicity of 
$
N^{\rm b}(K,\mu,\lambda)
$.
\end{proof}

\begin{proof}[Proof of Proposition \ref{prop:monotoneg}]
Fix a potential $V$. For each $\lambda>0$ and $x\in K_\infty$, put
\begin{equation} 
W(\lambda, V, x) = ((\lambda- V(x))_+)^{d_s/2} G\left(\frac{1}{2}\log((\lambda-V(x))_+)\right)
\end{equation}
and 
\begin{equation}
W^{\rm b}(\lambda,V,x) =((\lambda- V^{\rm b}(x))_+)^{d_s/2} G\left(\frac{1}{2}\log((\lambda-V^{\rm b}(x))_+)\right).
\end{equation}
Observe that $W(\lambda, V, x)=W((\lambda-V(x))_+)$ and $W^{\rm b}(\lambda, V,x) = W((\lambda-V^{\rm b}(x))_+)$. 
Using Proposition \ref{prop:monotone} we deduce the following two consequences. First, $\lambda\mapsto W(\lambda,V,x)$ is nonnegative and monotone nondecreasing for each $x$. And since $g(V,\lambda)$ is the weighted integral of $W(\lambda,V,x)$ over $x$, it follows that $\lambda\mapsto g(V,\lambda)$ is also monotone nondecreasing. The monotonicity of $\lambda\mapsto g^{\rm b}(V,\lambda)$ is proved in exactly the same way. Second, the monotonicity of $W(\lambda)$ implies that $W^\wedge(\lambda,V,x)\leq W(\lambda,V,x) \leq W^\vee(\lambda,V,x)$ for each $x$, and upon integration over $x$ we get $g^\wedge(V,\lambda)\leq g(V,\lambda) \leq g^\vee(V,\lambda)$.
\end{proof}

%
%

\subsection{Bohr's asymptotics via Dirichlet-Neumann bracketing}

We have all the necessary pieces to state the error of approximating $N(K_\infty,\mu_\infty, V,\lambda)$ by $g(V,\lambda)$.

\begin{theorem}[Error estimate in Bohr's approximation]
\label{thm:BohrError}
Under Assumptions \ref{ass:Weyl} and \ref{ass:sa}, we have
\begin{equation}
\label{eq:bohrcomp}
\left|\frac{N(K_\infty,\mu_\infty,V,\lambda)}{g(V,\lambda)} -1\right| \leq \max_{{\rm b}\in\{\wedge,\vee\}} \left|\frac{g^{\tilde{\rm b}}(V,\lambda)}{g^{\rm b}(V,\lambda)} -1+ \frac{\mathcal{R}^{\tilde{\rm b}}(V,\lambda)}{g^{\rm b}(V,\lambda)}\right|,
\end{equation}
where $\tilde{\rm b} = \wedge$ (resp. $\tilde{\rm b}=\vee$) if ${\rm b}=\vee$ (resp. if ${\rm b}=\wedge$).
\end{theorem}
\begin{proof}
From (\ref{eq:DNbracket}) we have
\begin{equation}
N^\wedge(K_\infty, \mu_\infty,V,\lambda)\leq N(K_\infty, \mu_\infty,V,\lambda) \leq N^\vee(K_\infty, \mu_\infty,V,\lambda).
\end{equation}
Meanwhile by Proposition \ref{prop:monotoneg},
\begin{equation}
g^\wedge(V,\lambda) \leq g(V,\lambda) \leq g^\vee(V,\lambda).
\end{equation}
Therefore
\begin{equation}
\label{eq:DNineq}
\frac{N^\wedge(K_\infty, \mu_\infty,V,\lambda)}{g^\vee(V,\lambda)} \leq \frac{N(K_\infty, \mu_\infty,V,\lambda)}{g(V,\lambda)} \leq \frac{N^\vee(K_\infty, \mu_\infty,V,\lambda)}{g^\wedge(V,\lambda)}.
\end{equation}
Subtract $1$ from every term in the inequality (\ref{eq:DNineq}), and then use (\ref{eq:NgR}) to replace $N^{\rm b}(K_\infty,\mu_\infty,V,\lambda)$ with $g^{\rm b}(V,\lambda) + \mathcal{R}^b(V,\lambda)$. Finally, we can estimate the absolute value of the middle term of the inequality by the maximum of the absolute value on either side of the inequality.
\end{proof}

Having established the main error estimate, Theorem \ref{thm:BohrError}, we can now give an abstract condition on $V$ for which Bohr's formula holds.

\begin{assumption}
\label{ass:pot}
The potential $V$ on $K_\infty$ satisfies 
\begin{equation}
\frac{g^\vee(V,\lambda)}{g^\wedge(V,\lambda)} = 1+o(1) \quad \text{as}~\lambda\to\infty.
\end{equation}
\end{assumption}

\begin{theorem}[Strong Bohr's formula]
\label{thm:Bohr}
Under Assumptions \ref{ass:Weyl}, \ref{ass:sa}, and \ref{ass:pot}, we have
\begin{equation}
\label{eq:strongBohr}
\lim_{\lambda\to\infty} \frac{N(K_\infty,\mu_\infty,V,\lambda)}{g(V,\lambda)} =1.
\end{equation}
\end{theorem}

\begin{proof}[Proof of Theorem \ref{thm:Bohr}]
Observe that Assumptions \ref{ass:Weyl} and \ref{ass:pot} together imply that the error term stated in Theorem \ref{thm:BohrError} is $o(1)$.
\end{proof}

\section{Connection between Bohr's formula and the distribution function of the potential} \label{sec:Bohrpot}

Assumption \ref{ass:pot} can be too abstract for applications dealing with fractal spaces. We now explain how this assumption can be restated in terms of distribution functions of $V$:
\begin{equation}
F(V,\lambda) := \mu_\infty(\{x \in K_\infty: V(x)\leq \lambda\}) \quad \text{and}\quad F^{\rm b}(V,\lambda) : =\mu_\infty(\{x\in K_\infty: V^{\rm b}(x) \leq \lambda\}.
\end{equation} 

\begin{lemma} We have that
\begin{equation}
g(V,\lambda) = \int_0^{W(\lambda)} \,F(V,\lambda-W^{-1}(t))\,dt \quad \text{and} \quad g^{\rm b}(V,\lambda) = \int_0^{W(\lambda)} \, F^{\rm b}(V,\lambda-W^{-1}(t))\,dt,
\end{equation}
where 
\begin{equation}
W^{-1}(t) = \inf\{\lambda \geq 0: W(\lambda) \geq t\}
\end{equation}
 is the generalized inverse of $W(\lambda) = \lambda^{d_s/2} G(\frac{1}{2}\log \lambda)$.
\end{lemma}

\begin{proof}
We start with a fundamental identity in measure theory. For any nonnegative function $f$ on a $\sigma$-finite measure space $(X,m)$, Fubini's theorem tells us that
\begin{equation}
\int_X \, f(x) \,m(dx) = \int_0^\infty \, m(\{x\in X: f(x) \geq t\})\,dt.
\end{equation}
Applying this identity to $g(V,\lambda)$ we find
\begin{equation}
\label{eq:gmeas}
g(V,\lambda) = \int_{K_\infty} W((\lambda-V(x))_+) \, d\mu_\infty(x) = \int_0^\infty\, \mu_\infty(\{x\in K_\infty: W((\lambda-V(x))_+) \geq t\})\, dt.
\end{equation}
Since $W$ is monotone nondecreasing (Proposition \ref{prop:monotone}), it has a well-defined generalized inverse $W^{-1}$, which satisfies
\begin{equation}
W(\lambda) \geq t \Longleftrightarrow \lambda \geq W^{-1}(t).
\end{equation}
So the right-hand term in (\ref{eq:gmeas}) can be further rewritten as
\begin{equation}
\int_0^\infty \, \mu_\infty\left(\{x\in K_\infty: (\lambda-V(x))_+ \geq W^{-1}(t)\}\right)\,dt.
\end{equation}
Now by assumption $V$ is a nonnegative potential, so $W^{-1}(t) \leq (\lambda-V(x))_+ \leq \lambda$, or equivalently, $t \leq W(\lambda)$. This places an upper bound on the integral, and we get
\begin{equation}
g(V,\lambda)=\int_0^{W(\lambda)} \, \mu_\infty\left(\{x\in K_\infty: V(x) \leq \lambda-W^{-1}(t)\}\right)\,dt = \int_0^{W(\lambda)} \, F(V, \lambda-W^{-1}(t))\,dt.
\end{equation}
The proof for $g^{\rm b}(V,\lambda)$ is identical.
\end{proof}

Observe that for $\lambda \leq \lambda'$,
\begin{eqnarray}
\nonumber 
g^\vee(V,\lambda) - g^\wedge(V,\lambda') &=& \int_0^{W(\lambda)} \, \left[F^\vee(V,\lambda-W^{-1}(t))- F^\wedge(V,\lambda'-W^{-1}(t))\right]\,dt\\
&&- \int_{W(\lambda)}^{W(\lambda')} \, F^\wedge(V,\lambda'-W^{-1}(t))\,dt,
\label{eq:gdiff1}
\end{eqnarray}
and
\begin{eqnarray}
\nonumber 
g^\vee(V,\lambda') - g^\wedge(V,\lambda) &=& \int_0^{W(\lambda)} \, \left[F^\vee(V,\lambda'-W^{-1}(t))- F^\wedge(V,\lambda-W^{-1}(t))\right]\,dt\\
&&+ \int_{W(\lambda)}^{W(\lambda')} \, F^\vee(V,\lambda'-W^{-1}(t))\,dt.
\label{eq:gdiff2}
\end{eqnarray}
These identities suggest that if the difference of the distribution functions $F^\vee(V,\lambda)-F^\wedge(V,\lambda)$ can be controlled, then one can control the difference $g^\vee(V,\lambda)-g^\wedge(V,\lambda)$. Indeed we have

\begin{proposition}
\label{prop:gdiffmeas}
Assumption \ref{ass:V} implies Assumption \ref{ass:pot}. Therefore, the strong Bohr's formula (Theorem \ref{thm:Bohr}) holds under Assumptions \ref{ass:Weyl}, \ref{ass:sa}, and \ref{ass:V}.
\end{proposition}
\begin{proof}
Let $h(V,\lambda) = \frac{F^\vee(V,\lambda)}{F^\wedge(V,\lambda)}-1 \geq 0$. Then
\begin{eqnarray}
0 &\leq & g^\vee(V,\lambda) - g^\wedge(V,\lambda) \\ &=& \int_0^{W(\lambda)} \, [1+ h(V,\lambda-W^{-1}(t))-1] F^\wedge(V,\lambda-W^{-1}(t))\,dt\\
&\leq&\left( \sup_{0\leq t \leq W(\lambda)} h(V,\lambda-W^{-1}(t))\right) \int_0^{W(\lambda)} \, F^\wedge(V,\lambda-W^{-1}(t))\,dt\\
&=& \left(\sup_{0\leq s\leq \lambda} h(V,s) \right) g^\wedge(V,\lambda).
\end{eqnarray}
Assumption \ref{ass:V} implies that $\sup_{0\leq s\leq \lambda} h(V,s) = o(1)$ as $\lambda\to\infty$, so we obtain Assumption \ref{ass:pot}.
\end{proof}

\subsection{A weak version of Bohr's formula}

Motivated by \cites{StrichartzSHO,OSt,OS-CT}, we also give a weak version of Bohr's formula as follows.

\begin{theorem}[Weak Bohr's formula]
\label{thm:Bohr-w}
Let $\lambda^*>\lambda$ with $\lambda^*-\lambda=o(\lambda)$ and
\begin{equation}
\label{eq:Fasymp}
\frac{F^\vee(V,\lambda)}{F^\wedge(V,\lambda^*)} = 1+o(1) \quad\text{and}\quad \frac{F^\wedge(V,\lambda)}{F^\vee(V,\lambda^*)} = 1+o(1).
\end{equation}
Then, with Assumptions \ref{ass:Weyl} and \ref{ass:sa}, we have
\begin{equation}
\label{eq:Bohr-w}
\lim_{\lambda\to\infty} \frac{N(K_\infty,\mu_\infty,V,\lambda)}{g(V,\lambda^*)} =1,
\end{equation}
\end{theorem}

The statement of Theorem \ref{thm:Bohr-w} is reminiscent of the situation when one compares two nondecreasing distribution jump functions with closely spaced jumps. When the jumps asymptotically coincide, then the 
 difference of corresponding measures tends to zero in the sense of weak convergence. 
 
\begin{proof}
By mimicking the proof of Theorem \ref{thm:BohrError} we get
\begin{equation}
\left|\frac{N(K_\infty,\mu_\infty,V,\lambda)}{g(V,\lambda^*)} -1\right| \leq \max_{{\rm b}\in\{\wedge,\vee\}} \left|\frac{g^{\tilde{\rm b}}(V,\lambda)}{g^{\rm b}(V,\lambda^*)} -1+ \frac{\mathcal{R}^{\tilde{\rm b}}(V,\lambda)}{g^{\rm b}(V,\lambda^*)}\right|.
\end{equation}
Since $\lambda^* -\lambda = o(\lambda)$ as $\lambda\to\infty$, the ratio $\mathcal{R}^{\tilde{\rm b}}(V,\lambda)/g^{\rm b}(V,\lambda^*)$ can be made to be $o(1)$. So the key estimate is to show that $g^{\tilde{\rm b}}(V,\lambda)/g^{\rm b}(V,\lambda^*) = 1+o(1)$ for \emph{both} ${\rm b}=\wedge$ \emph{and} ${\rm b}=\vee$. (This is to contrast with the case $\lambda'=\lambda$ as shown in Proposition \ref{prop:gdiffmeas}, where a one-sided bound suffices because $g^\vee(V,\lambda) - g^\wedge(V,\lambda) \geq 0$.)

From (\ref{eq:gdiff1}) we find
\begin{eqnarray}
|g^\vee(V,\lambda)-g^\wedge(V,\lambda^*)| &\leq& W(\lambda) \left(\sup_{0\leq s \leq \lambda} [F^\vee(V,s) - F^\wedge(V,s+\lambda^*-\lambda)]]\right)\\
&& + [W(\lambda^*)-W(\lambda)]\left(\sup_{0\leq s \leq \lambda^*-\lambda} F^\wedge(V,s)\right)
\end{eqnarray}
According to the first condition in (\ref{eq:Fasymp}), $\sup_{0\leq s \leq \lambda} [F^\vee(V,s) - F^\wedge(V,s+\lambda^*-\lambda)]] = o(F^\vee(V,\lambda))$ and $\sup_{0\leq s \leq \lambda^*-\lambda} F^\wedge(V,s) = o(F^\wedge(V,\lambda^*))$. This implies that the absolute value on the RHS of (\ref{eq:bohrcomp}) is $o(1)$ for ${\rm b} =\wedge$. Similarly, the second condition in (\ref{eq:Fasymp}) implies that the absolute value on the RHS of (\ref{eq:bohrcomp}) is $o(1)$ for ${\rm b}= \vee$ also.
\end{proof}

\section{Laplace transform (heat kernel trace) version of Bohr's formula} \label{sec:LaplaceBohr}

In this section we impose Assumption \ref{ass:sa} and either one of Assumptions \ref{ass:SpecHKT} and \ref{ass:WeylHKT}, and prove Theorems \ref{thm:specdimHKT} and \ref{thm:LaplaceBohr}. Let us introduce the traces
\begin{eqnarray}
\mathcal{L}(K_\infty,\mu_\infty,V,t) &:=& {\rm Tr}_{K_\infty}\{e^{-t(-\Delta+V)}\},\\
\mathcal{L}^\wedge(K_\infty,\mu_\infty,V,t) &:=& {\rm Tr}_{K_\infty}\{e^{-t(-\Delta^\wedge+V^\wedge)}\},\\
\mathcal{L}^\vee(K_\infty,\mu_\infty,V,t) &:=& {\rm Tr}_{K_\infty}\{e^{-t(-\Delta^\vee+V^\vee)}\}.
\end{eqnarray}
Observe that $ \mathcal{L}^\wedge(K_\infty,\mu_\infty,V,t) \leq\mathcal{L}(K_\infty,\mu_\infty,V,t) \leq \mathcal{L}^\vee(K_\infty,\mu_\infty,V,t) $.

Since $L^2(K_\infty,\mu_\infty) = \bigoplus_\alpha L^2(K_\alpha,\mu_\alpha)$, it follows that
\begin{equation}
\label{eq:HKTsum}
\mathcal{L}^{\rm b}(K_\infty,\mu_\infty, V, t) = \sum_\alpha \mathcal{L}^{\rm b}(K_\alpha, \mu_\alpha, V, t),\quad {\rm b} \in \{\wedge, \vee\},
\end{equation}
where
\begin{equation}
\mathcal{L}^{\rm b}(K_\alpha, \mu_\alpha, V,t) = {\rm Tr}_{K_\alpha}\left\{e^{-t(-\Delta^{\rm b}+V^{\rm b})}\right\}= \mathcal{L}^{\rm b}(K_\alpha, \mu_\alpha, t) \cdot \exp\left(-t \left.V^{\rm b}\right|_{K_\alpha}\right).
\end{equation}
Let
\begin{equation}
\mathcal{F}(V,t) :=\int\limits_{K_\infty} e^{-tV(x)}\, \mu_\infty(dx).
\end{equation}
Similarly define
\begin{equation}
\mathcal{F}^{\rm b}(V,t) := \int\limits_{K_\infty} e^{-tV^{\rm b}(x)}\, \mu_\infty(dx)
\end{equation}
for ${\rm b}\in \{\wedge, \vee\}$. Observe that $\mathcal{F}^\wedge(V,t) \leq \mathcal{F}(V,t) \leq \mathcal{F}^\vee(V,t)$, and that  Assumption \ref{ass:sa} ensures that $\mathcal{F}(V,t)$ and $\mathcal{F}^{\rm b}(V,t)$ are finite for $t>0$.

\begin{proof}[Proof of Theorem \ref{thm:specdimHKT}]
Let us first note that
\begin{equation}
\label{eq:LF}
\frac{\mathcal{L}^\wedge(K_\infty,\mu_\infty,V,t)}{t^{-d_s/2} \mathcal{F}^\vee(V,t)}\leq \frac{\mathcal{L}(K_\infty,\mu_\infty,V,t)}{t^{-d_s/2} \mathcal{F}(V,t)} \leq \frac{\mathcal{L}^\vee(K_\infty,\mu_\infty,V,t)}{t^{-d_s/2} \mathcal{F}^\wedge(V,t)}.
\end{equation}
By (\ref{eq:HKTsum}),
\begin{align}
\mathcal{L}^{\rm b}(K_\infty,\mu_\infty,V,t) &= \sum_\alpha \mathcal{L}^{\rm b}(K_\alpha, \mu_\alpha, V,t) \\
&= \sum_\alpha \mathcal{L}^{\rm b}(K_\alpha, \mu_\alpha, t) \cdot \exp\left(-t \left.V^{\rm b}\right|_{K_\alpha}\right) \\ 
&= \sum_\alpha \mathcal{L}^{\rm b}(K_\alpha, \mu_\alpha, t) \cdot \int\limits_{K_\alpha} e^{-t V^{\rm b}(x)} \, \mu_\alpha(dx).
\label{eq:Lb}
\end{align}
Under Assumption \ref{ass:SpecHKT}, there exist positive constants $C_1$ and $C_2$ such that for all sufficiently small $t$,
\begin{eqnarray}
t^{d_s/2} \mathcal{L}^\vee(K_\infty,\mu_\infty,V,t) &\leq& C_1\sum_\alpha \int_{K_\alpha} \,e^{-tV^\vee(x)}\,\mu_\alpha (dx) = C_1 \mathcal{F}^\vee(V,t), \\
t^{d_s/2} \mathcal{L}^\wedge(K_\infty,\mu_\infty,V,t) &\geq& C_2 \sum_\alpha \int_{K_\alpha}\, e^{-tV^\wedge(x)} \, \mu_\alpha(dx) = C_2 \mathcal{F}^\wedge(V,t).
\end{eqnarray}
Meanwhile, by Fubini's theorem and by the nonnegativity of $V$, we have
\begin{eqnarray}
\mathcal{F}^{\rm b}(V,t) &=& \int_0^\infty \mu_\infty\left(\{x\in K_\infty: e^{-tV^{\rm b}(x)} \geq s\} \right)\,ds\\
&=& \int_{-\infty}^\infty \, \mu_\infty\left(\{x\in K_\infty: e^{-tV^{\rm b}(x)} \geq e^{-t\lambda} \right) te^{-t\lambda} \,d\lambda \\
&=& \int_0^\infty \, \mu_\infty\left(\{x \in K_\infty: V^{\rm b}(x) \leq \lambda\}\right) te^{-t\lambda} \, d\lambda\\
&=& \int_0^\infty \, F^{\rm b}(V,\lambda) t e^{-t\lambda}\,d\lambda.
\end{eqnarray}
Hence under Assumption \ref{ass:V-s}, there exists $\lambda_0>0$ such that
\begin{eqnarray}
\label{eFVt}
\mathcal{F}^\vee(V,t) &=& \int_0^\infty F^\vee(V,\lambda) t e^{-t\lambda}\,d\lambda\\
&=& \int_0^{\lambda_0} F^\vee(V,\lambda)t e^{-t\lambda}\,d\lambda + \int_{\lambda_0}^\infty F^\vee(V,\lambda)te^{-t\lambda}\,d\lambda\\
&\leq& F^\vee(V,\lambda_0) \int_0^{\lambda_0} t e^{-t\lambda}\,d\lambda + C \int_{\lambda_0}^\infty F^\wedge\left(V,\frac{\lambda}{2}\right) t e^{-t\lambda} \, d\lambda\\
&=& F^\vee(V,\lambda_0) \left(1-e^{-t\lambda_0}\right) + C \int_{\lambda_0/2}^\infty \, F^\wedge(V,\lambda) \cdot 2te^{-2t\lambda}\,d\lambda\\
&\leq& F^\vee(V,\lambda_0) \left(1-e^{-t\lambda_0}\right) + C \mathcal{F}^\wedge(V,2t).
\end{eqnarray}
Therefore
\begin{equation}
\frac{\mathcal{F}^\vee(V,t)}{\mathcal{F}^\wedge(V,t)} \leq \frac{\mathcal{F}^\vee(V,t)}{\mathcal{F}^\wedge(V,2t)} \leq C+ F^\vee(V,\lambda_0) \frac{1-e^{-t\lambda_0}}{\mathcal{F}^\wedge(V,2t)}.
\end{equation}
Since $\lim_{t\downarrow 0} (1-e^{-t\lambda_0})=0$ and $t\mapsto \mathcal{F}^\wedge(V,2t)$ is monotone decreasing, it follows that
\begin{equation}
\varlimsup_{t\downarrow 0} \frac{\mathcal{F}^\vee(V,t)}{\mathcal{F}^\wedge(V,t)} \leq C+F^\vee(V,\lambda_0) \varlimsup_{t\downarrow 0} \frac{1-e^{-t\lambda_0}}{\mathcal{F}^\wedge(V,2t)} = C.
\end{equation}
Putting everything together we find
\begin{eqnarray}
\varlimsup_{t\downarrow 0} \frac{\mathcal{L}(K_\infty,\mu_\infty,V,t)}{t^{-d_s/2} \mathcal{F}(V,t)} &\leq& \left(\varlimsup_{t\downarrow 0} \frac{t^{d_s/2} \mathcal{L}^\vee(K_\infty,\mu_\infty,V,t)}{\mathcal{F}^\vee(V,t)}\right) \left( \varlimsup_{t\downarrow 0} \frac{\mathcal{F}^\vee(V,t)}{\mathcal{F}^\wedge(V,t)}\right),  \\
\varliminf_{t\downarrow 0} \frac{\mathcal{L}(K_\infty,\mu_\infty,V,t)}{t^{-d_s/2} \mathcal{F}(V,t)} &\geq& \left(\varliminf_{t\downarrow 0} \frac{t^{d_s/2} \mathcal{L}^\wedge(K_\infty,\mu_\infty,V,t)}{\mathcal{F}^\wedge(V,t)}\right) \left( \varliminf_{t\downarrow 0} \frac{\mathcal{F}^\wedge(V,t)}{\mathcal{F}^\vee(V,t)}\right).
\end{eqnarray}
Thus
\begin{equation}
\label{eq:LFfinal}
C_2 C^{-1}\leq \varliminf_{t\downarrow 0} \frac{\mathcal{L}(K_\infty,\mu_\infty,V,t)}{t^{-d_s/2} \mathcal{F}(V,t)}\leq \varlimsup_{t\downarrow 0} \frac{\mathcal{L}(K_\infty,\mu_\infty,V,t)}{t^{-d_s/2} \mathcal{F}(V,t)} \leq C_1 C.
\end{equation}

Finally, regarding the spectral dimension of $-\Delta+V$, we note that $F(V,\lambda)=\Theta(\lambda^\beta)_{\lambda\to\infty}$ is equivalent to $\mathcal{F}(V,t) = \Theta(t^{-\beta})_{t\downarrow 0}$, an easy consequence of Laplace transform. Thus according to (\ref{eq:LFfinal}), $\mathcal{L}(K_\infty,\mu_\infty,V,t) \asymp t^{-(d_s+2\beta)/2}$ as $t\downarrow 0$.
\end{proof}

\begin{proof}[Proof of Theorem \ref{thm:LaplaceBohr}]
Combining (\ref{eq:Lb}) with Assumption \ref{ass:WeylHKT} we obtain 
\begin{equation} 
\label{eq:LF2}
 \mathcal{L}^{\rm b}(K_\infty,\mu_\infty,V,t) = t^{-d_s/2}\left[ H(t)+\rho^{\rm b}(t)\right] \mathcal{F}^{\rm b}(V,t) 
\end{equation}
which, together with  (\ref{eq:LF}) and after  some manipulation, implies
\begin{equation}
\label{eq:LBerror}
\left|\frac{\mathcal{L}(K_\infty,\mu_\infty,V,t)}{t^{-d_s/2}H(t) \mathcal{F}(V,t)}-1\right| \leq \max_{{\rm b}\in\{\wedge,\vee\}} \left| \left(1+\frac{\rho^{\tilde{\rm b}}(t)}{H(t)}\right) \frac{\mathcal{F}^{\tilde{\rm b}}(V,t)}{\mathcal{F}^{\rm b}(V,t)} -1\right|.
\end{equation}
Next, by Assumption \ref{ass:V} and \eqref{eFVt}, 
\begin{eqnarray}
\mathcal{F}^\vee(V,t) &=& 
\mathcal{F}^\wedge(V,t) +
  \int_0^\infty \, \left(F^\vee(V,\lambda) -F^\wedge(V,\lambda)\right) \,te^{-t\lambda}\,d\lambda\\
&=& \mathcal{F}^\wedge(V,t) + o\left(\mathcal{F}^\wedge(V,t)\right)_{t\downarrow 0}
\end{eqnarray}
because for any $\delta>0$ there is $\lambda_\delta>0$ such that 
$ 
F^\vee(V,\lambda) -F^\wedge(V,\lambda)
<\delta
F^\vee(V,\lambda) 
$ when $\lambda>\lambda_\delta$.
Thus $\frac{\mathcal{F}^\vee(V,t) }{\mathcal{F}^\wedge(V,t)} = 1+o(1)$ as $t\downarrow 0$. Hence (\ref{eq:LBerror}) implies (\ref{eq:LaplaceBohr}).
\end{proof}


\section{Examples} \label{sec:examples}
 
In this section we provide several instances on both classical and fractal settings whereby the existence of the spectral dimension of $-\Delta+V$ can be proved, and moreover, Bohr's formula holds.

\subsection{Euclidean spaces}

One would be remiss not to mention the most classical setting, which is the Schr\"odinger operator $-\Delta+V$ on $\mathbb{R}^d$, where $\Delta =\sum_{i=1}^d (\partial^2/\partial x_i^2)$ and $V$ is an unbounded potential. See \emph{e.g.} \cite{ReedSimonVol4}*{\S XIII.15}. The key idea is to partition $\mathbb{R}^d$ (the unbounded space $K_\infty$) into cubes of side $1$ (the cells $K_\alpha$). Then by applying the machinery outlined in the previous section, one arrives at the following well-known result: if $V(x)=\Theta(|x|^\beta)$ as $|x|\to \infty$, then Bohr's formula holds, and the spectral dimension of this Schr\"odinger operator is $d(1+2/\beta)$.

In dimension $1$ Bohr's formula can be established for logarithmically diverging potentials. The proof method involves solving a Sturm-Liouville ODE, which appears rather particular to one-dimensional settings, and may be difficult to generalize to higher dimensions. We refer the reader to \cites{NaimarkSolomyak, HoltMolchanov} for more details. 

%
%

\subsection{Infinite fractafolds based on nested fractals} \label{sec:fractafold}


%
%

Nested fractals are introduced first by Lindstr\o m \cite{Lindstrom}. The typical examples to keep in mind are the Sierpinski gaskets $SG(n)$, where $n$ denotes the length scale of the subdivision. There are also higher-dimensional analogs of $SG$.

On nested fractals, and more generally post-critically finite (p.c.f.) fractals, one can define a notion of the Laplacian (or a Brownian motion). See \emph{e.g.} \cite{BarlowStFlour}*{\S 2$\sim$\S4}, \cite{Kigami}*{Chapters 2$\sim$3}, \cite{StrBook}*{Chapters 1$\sim$2} for the relevant definitions and results. We will need just one result on the spectral asymptotics of the Laplacian on p.c.f. fractals with regular harmonic structure.

\begin{proposition}[\cite{KigamiLapidus}*{Theorem 2.4},~\cite{Kigami}*{Theorem 4.1.5}]
\label{prop:ECFspecKigami}
Let $K$ be a p.c.f. fractal, and $\mu$ be a self-similar measure on $K$ with weight $(\mu_i)_{i=1}^N$. Assume that $\mu_i r_i <1$ for all $i\in \{1,2,\cdots, N\}$. Let $d_s$ be the unique number $d$ which satisfies $\sum_{i=1}^N \gamma_i^d=1$, where $\gamma_i = \sqrt{r_i \mu_i}$. Let $N^\wedge(K,\mu,\lambda)$ (resp. $N^\vee(K,\mu,\lambda)$) be the eigenvalue counting function for the Laplacian on $L^2(K,\mu)$ with Dirichlet (resp. Neumann) boundary condition. Then for ${\rm b}\in \{\wedge,\vee\}$,
\begin{equation}
0<\varliminf_{\lambda \to\infty} \lambda^{-d_s/2} N^{\rm b}(K,\mu,\lambda) \leq \varlimsup_{\lambda\to\infty} \lambda^{-d_s/2} N^{\rm b}(K,\mu,\lambda) <\infty.
\end{equation}
Moreover:
\begin{enumerate}[label=(\alph*)]
\item Non-lattice case: If $\sum_{i=1}^N \mathbb{Z}\log \gamma_i$ is a dense subgroup of $\mathbb{R}$, then the limit $$\lim_{\lambda\to\infty} \lambda^{-d_s/2} N^{\rm b}(K,\mu,\lambda)$$ exists, and is independent of the boundary conditions.
\item Lattice case: If $\sum_{i=1}^N \mathbb{Z} \log \gamma_i$ is a discrete subgroup of $\mathbb{R}$, let $T>0$ be its generator. Then as $\lambda\to\infty$,
\begin{equation}
N^{\rm b}(K,\mu,\lambda) = \left[G\left(\frac{\log \lambda}{2}\right) + o(1)\right] \lambda^{d_s/2},
\end{equation}
where $G$ is a right-continuous, $T$-periodic function with $0< \inf G \leq \sup G <\infty$, and is independent of the boundary conditions.
\end{enumerate}
\end{proposition}

We remark that the proof of Proposition \ref{prop:ECFspecKigami} relies upon Feller's renewal theorem~\cite{FellerVol2}.

Our goal is to state Bohr's formula for the Schr\"odinger operator on a class of unbounded fractal spaces. As mentioned in Section \ref{sec:Intro}, one such candidate is a fractafold based on a nested fractal. We shall consider two types:

\begin{enumerate}[label=(\roman*)]
\item The infinite blow-ups of a nested fractal in $\mathbb{R}^d$, $d\geq 2$. (See Figure \ref{fig:InfiniteBlowup}).  
\item Infinite \textbf{periodic fractafolds} $K_\infty$ based on the planar Sierpinski gasket $K=SG(n)$, equipped with a metric $R$. (In practice, $R$ is taken to be the resistance metric, but the results to follow do not depend explicitly on the specifics of $R$.) The examples we will consider can be constructed by first defining an infinite ``cell graph'' $\Gamma$, and then replacing each vertex of $\Gamma$ by a copy of $K$, and gluing the $K_\alpha$ in a consistent 
way. With this construction the metric $R$ on $K$ extends to a metric $R$ on $K_\infty$ in the obvious way. For instance, one can construct the ladder periodic fractafold (Figure \ref{fig:LadderFractafold}) and the hexagonal periodic fractafold (Figure \ref{fig:HexagonalFractafold}).
\end{enumerate}

To establish Bohr's formula, we will need information about the measure growth of balls in $K_\infty$. For the infinite blow-ups of a nested fractal, it is direct to verify that for all $x\in K_\infty$ and $r>0$,
\begin{equation}
\label{eq:blowupmeas}
c r^{d_{f,R}} \leq \mu_\infty(B_R(x,r)) \leq C r^{d_{f,R}},
\end{equation}
where $d_{f,R}$ is the Hausdorff dimension of $K$ with respect to the metric $R$ on $K$.

For the periodic fractafolds a slightly different analysis is needed. Let $d_\Gamma$ be the graph metric of the cell graph $\Gamma$, and $B_{d_\Gamma}(z,r) := \{y\in \Gamma : d_G(z,y) \leq r\}$ be the ball of radius $r$ centered at $z$ in $\Gamma$. Since $K_\infty$ is constructed by replacing each vertex of $\Gamma$ by a copy of $K$, we can estimate the volume growth of balls in $K_\infty$ using the cardinality of balls in $\Gamma$.
\begin{proposition}
\label{prop:cover}
Let $D(K) := {\rm diam}_R(K)$. For all $x\in K_\infty$ and all $r> 2D(K)$,
\begin{equation} 
\label{eq:cover}
\mu_\infty\left(B_{d_\Gamma}\left(\psi(x),~r-2D(K)\right)\right) \leq \mu_\infty(B_R(x,r)) \leq  \mu_\infty\left(B_{d_\Gamma}\left(\psi(x), ~r+2D(K)\right)\right),
\end{equation}
where $\psi(x)$ is the vertex in $\Gamma$ which is replaced by the cell $K_\alpha \ni x$ in the periodic fractafold construction.
\end{proposition}
\begin{proof}
Let $\eta(r) := r/D(K)>2$. Then $B_R(x,r) = B_R(x,\eta(r) D(K))$ and
\begin{equation}
 B_R(y,(\lfloor\eta(r)\rfloor-1) D(K)) \subseteq  B_R(x,\eta(r) D(K)) \subseteq  B_R(y,(\lceil\eta(r)\rceil+1) D(K))
\end{equation}
for any $y$ which lies in the same cell $K_\alpha$ as $x$. Here $\lfloor \alpha \rfloor$ (resp. $\lceil \alpha \rceil$) denotes the largest integer less than or equal to $\alpha$ (resp. the smallest integer greater than or equal to $\alpha$). It is then direct to show that there exist $y$ such that $B_R(y,(\lceil\eta(r)\rceil+1)D(K))$ is covered by the union of all cells $K_\alpha$ which are at most distance $(\lceil\eta(r)\rceil+1)$ from $y$ in the $\Gamma$ metric. Since each cell has $\mu$-measure $1$, the $\mu$-measure of the cover is equal to the cardinality of $B_{d_\Gamma}(\psi(x),\lceil \eta(r)\rceil+1)$. The upper bound in (\ref{eq:cover}) follows by overestimating $\lceil\eta(r)\rceil+1$ by $\eta(r)+2$. The proof of the lower bound is similar.
\end{proof}

We can now state the main result of this subsection.

\begin{proposition}
\label{prop:SGV}
On the infinite blow-up of a nested fractal (resp. the ladder periodic fractafold based on $SG(n)$, the hexagonal periodic fractafold based on $SG(n)$), Bohr's formula holds for potential of the form $V(x) \sim R(0,x)^\beta$ for any $\beta>0$. In particular, the spectral dimension of $-\Delta+ V$ is $d_s(V) = d_s+ 2 (d_h/\beta)$, where $d_h$ equals the Hausdorff dimension of the nested fractal with respect to the metric $R$ (resp. $1$, $2$).
\end{proposition} 
\begin{proof}
Since each $K_\alpha$ which makes up the cellular decomposition of $K_\infty$ is isometric to the same nested fractal $K$, by Proposition \ref{prop:ECFspecKigami} we have that Assumption \ref{ass:Weyl} holds.

Because the cells $K_\alpha$ intersect at boundary points in a natural way, the Dirichlet form $(\mathcal{E}, \mathcal{F})$ corresponding to the Laplacian $-\Delta$ on $L^2(K_\infty,\mu_\infty)$ can be built up as a sum of the constituent Dirichlet forms on $L^2(K_\alpha,\mu_\alpha)$. Hence one can show self-adjointness of $-\Delta$ in the sense of quadratic forms. 
And since the potential $V(x)$ grows unboundedly as $d(0,x)\to +\infty$, Assumption \ref{ass:sa} then implies that $(-\Delta+V)$ has pure point spectrum.

 For condition (i), one can confirm that there exist constants $c$ and $C$ such that for all $x\in K_\infty$ and all sufficiently large $r>0$,
 \begin{equation}
 \label{eq:meas}
 cr^{d_h}\leq \mu_\infty(B_R(x,r)) \leq Cr^{d_h}.
 \end{equation}
 For the infinite blow-up (\ref{eq:meas}) follows from (\ref{eq:blowupmeas}) with $d_h = d_{h,R}$. As for the periodic fractafolds, note that the corresponding cell graphs $\Gamma$ satisfy
 \begin{equation}
| B_\Gamma(z,r) |\asymp r^{d_{h,\Gamma}} \quad \text{for all}~z \in \Gamma~\text{and}~r>0,
 \end{equation}
 where $d_{h,\Gamma}$ equals $1$ (resp. $2$) in the case of the ladder fractafold (resp. the hexagonal fractafold). Combining this with Proposition \ref{prop:cover} we get (\ref{eq:meas}) with $d_h=d_{h,\Gamma}$. In all cases, we the find
\begin{equation}
F(\lambda) = \mu_\infty(\{x: V(x)<\lambda\}) \simeq \mu_\infty(B_R(0,\lambda^{1/\beta})) \simeq \lambda^{d_{h,\Gamma}/\beta},
\end{equation}
and the same asymptotics holds for $F^\wedge(\lambda)$ and $F^\vee(\lambda)$. Finally, to see that condition (ii) holds, we use the inequality
\begin{equation}
[V^\wedge(x) - V^\vee(x)] \leq [R(0,x)+1]^\beta - [R(0,x)-1]^\beta \leq C_\beta [R(0,x)+1]^{\beta-1},
\end{equation}
where $C_\beta$ is an explicit constant depending on $\beta$ only. Observe that the RHS is uniformly bounded from above by a constant multiple of $\lambda^{1-\beta^{-1}}$ for all $x$ in the set $\{x:V^\vee(x)\leq \lambda\}$. 
%
\end{proof}

\subsection{Infinite fractal fields based on nested fractals}\label{fr-fi}

There is another notion of an unbounded space based on compact fractals, which are known as \textbf{fractal fields}. The name originated from Hambly and Kumagai \cite{HamKumFields}, who were interested in studying fractal penetrating Brownian motions. Fractal fields differ from the fractafolds of the previous subsection in that we do not require neighborhoods of (junction) points in $K_\infty$ to be homeomorphic to $K$.

First consider the triangular lattice finitely ramified Sierpinski fractal field introduced in \cite{StrTep}*{\S 6}, see Figure \ref{fig:TriFractalfield}. Notice that this fractal field admits a cellular decomposition whereby cells adjoin at boundary vertices of $SG(n)$. As a result, the proof strategy from the previous Proposition \ref{prop:SGV} applies in this setting. 


\begin{proposition}
On the triangular lattice finitely ramified fractal field based on $SG(n)$, Bohr's formula holds for potential of the form $V(x) \sim R(0,x)^\beta$ for any $\beta>0$. In particular, the spectral dimension of $-\Delta+ V$ is $d_s(V) = d_s+ (4/\beta)$.
\end{proposition}

\begin{figure}
\centering
\includegraphics[width=0.4\textwidth]{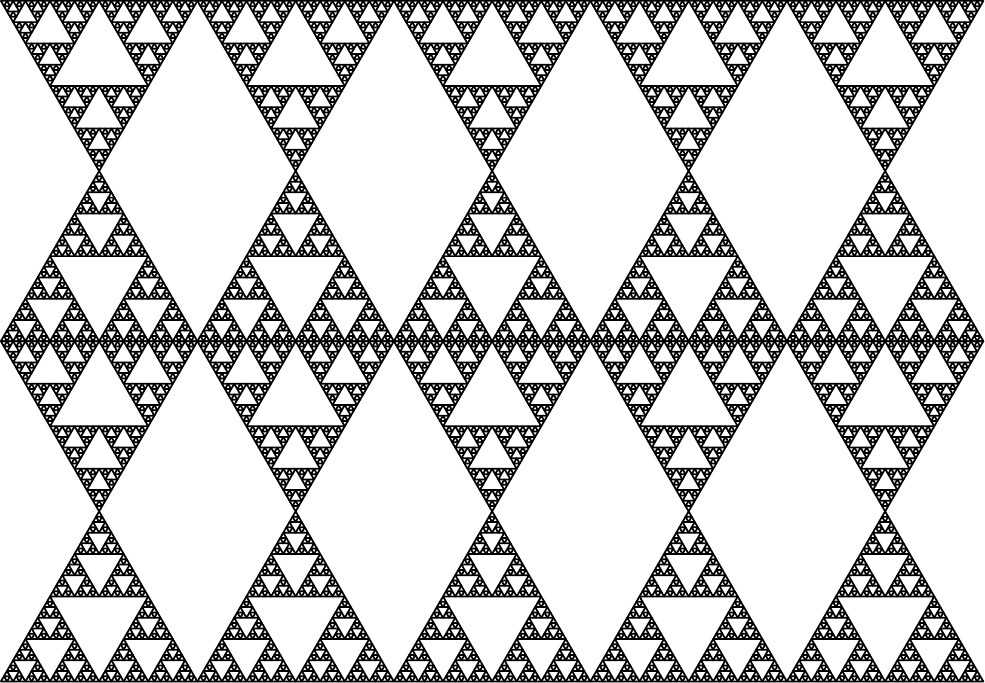}
\caption{The double-ladder fractal field based on $SG(2)$.}
\label{fig:doubleladder}
\end{figure}

Next we consider the double-ladder fractal field based on $SG(2)$, see Figure \ref{fig:doubleladder}. An important difference here is that pairs of $SG(2)$ cells may adjoin either at a point or along a boundary segment, which makes this space infinitely ramified. In order to analyze this example using our methods, one needs to understand the eigenvalue problem for the Laplacian on $SG(2)_\Omega:=SG(2) \setminus \partial\Omega$, where $\partial\Omega$ consists of the top vertex and the bottom edge of $SG(2)$. This was investigated by Qiu \cite{Qiu}, whose result we quote below.

\begin{proposition}[\cite{Qiu}*{Theorem 3.10}]
\label{prop:Qiu}
Let $N^{\rm b}_\Omega(\lambda)$ be the eigenvalue counting function for the Laplacian on $SG(2)_\Omega$ with boundary condition ${\rm b} \in \{\wedge, \vee\}$ on the top vertex and the bottom edge of $SG(2)$. Then there exists a c\`{a}dl\`{a}g $\log 5$-periodic function $G:\mathbb{R}\to\mathbb{R}$, with $0<\inf G <\sup G <\infty$ and independent of ${\rm b}$, such that
\begin{align}
N^{\rm b}_\Omega(\lambda)= G(\log \lambda) \lambda^{\log 3/\log 5} + O\left(\lambda^{\log 2/\log 5} \log\lambda\right)
\end{align}
 as $\lambda\to\infty$.
\end{proposition}

Using this result we can show the validity of Bohr's formula in this setting.
\begin{proposition}
On the double-ladder fractal field based on $SG(2)$, Bohr's formula holds for potential of the form $V(x) \sim R(0,x)^\beta$ for any $\beta>0$.
\end{proposition}
\begin{proof}
The one nontrivial assumption to check is Assumption \ref{ass:Weyl}, which is furnished by Proposition \ref{prop:Qiu}. The other two assumptions, \ref{ass:sa} and \ref{ass:V}, are verified easily. The result then follows from Theorem \ref{thm:Bohrmain}. 
\end{proof}
There are some obvious extensions of the double-ladder fractal field example, which we leave to the reader. An interesting open problem is to study the applicability of Bohr's formula to the original fractal field (or gasket tiling) in \cite{HamKumFields}, shown in Figure \ref{fig:SGFractalField}. We note that heat kernel estimates are established on this fractal field \cite{HamKumFields}*{Theorem 1.1}. However, to the best of the authors' knowledge, there is no corresponding Weyl asymptotic (or heat kernel trace asymptotic) estimate which is sharp to an $o(1)$ remainder. In particular, the fact that the $SG$ cells adjoin along edges rather than at points makes the analysis more delicate.

\begin{figure}
\centering
\includegraphics[width=0.4\textwidth]{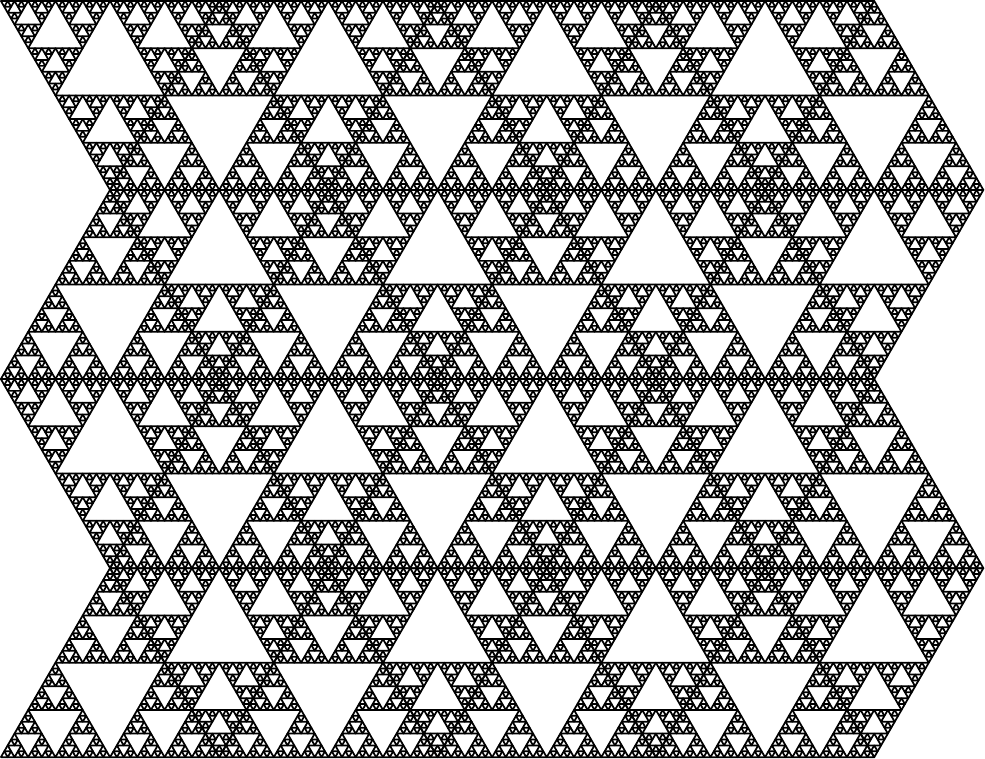}
\caption{The $SG(2)$ fractal field (or gasket tiling) considered in \cite{HamKumFields}.}
\label{fig:SGFractalField}
\end{figure}

\subsection{Infinite Sierpinski carpets}

Let $F\subset \mathbb{R}^d~(d\geq 2)$ be a generalized Sierpinski carpet in the sense of \cite{BarlowBass,BBKT}, and let $F_n$ be its $n$th-level approximation. Following \cite{BarlowBass}, we call $\tilde{F}=\bigcup_{n\in \mathbb{N}_0} \ell^n F_n$ the \emph{pre-carpet}, and $F_\infty = \bigcup_{n\in \mathbb{N}_0} \ell^n F$ the \emph{infinite carpet}. The difference between the two is that $\tilde{F}$ is tiled by unit squares and has nonzero Lebesgue measure, whereas $F_\infty$ is tiled by copies of the same Sierpinski carpet $F$ and has zero Lebesgue measure. In both cases, we adopt the Euclidean metric $|\cdot|$ and regard $(K_\infty, \mu_\infty, |\cdot|)$ as the metric measure space, which has volume growth
\begin{equation}
c_1 r^{d_f} \leq \mu_\infty(B(x,r)) \leq c_2 r^{d_f} \qquad (x\in K_\infty, r>0),
\end{equation}
where $d_f=(\log m/\log \ell)$ is the Hausdorff dimension of the carpet $F$ with respect to the Euclidean metric.

\begin{proposition}
Bohr's formula holds on the pre-carpet $\tilde{F}$ with potential $V(x) \sim |x|^\beta$ for any $\beta>0$. In particular, the spectral dimension of $(-\Delta+V)$ on $\tilde{F}$ is $d+2(d_f/\beta)$, where $d$ is the dimension of the ambient space $\mathbb{R}^d$ in which $\tilde{F}$ lies.
\end{proposition}

The case of the infinite carpet is more nuanced. Hambly \cite{Hamblyspec} and Kajino \cite{Kajinospec} proved that the heat kernel trace of the bare Laplacian on $F$ satisfies Assumption \ref{ass:WeylHKT}, with $H$ a continuous periodic function of $\log t$ (though it is NOT known whether $H$ is non-constant). Kajino \cite{Kajinospec2} further showed the asymptotics of the heat kernel trace to \emph{all} orders of the boundary terms. Note that their results imply that the eigenvalue counting function satisfies the asymptotics $c_1 \lambda^{d_s/2} \leq N^{\rm b}(F,\mu,\lambda) \leq c_2\lambda^{d_s/2}$, but do NOT necessarily imply the sharper estimate, Assumption \ref{ass:Weyl}. As mentioned earlier, this is because the classical techniques of Tauberian theorems cannot be applied here.

\begin{proposition}
The Laplace transform version of Bohr's formula holds on the infinite carpet $F_\infty$ with potential $V(x) \sim |x|^\beta$ for any $\beta>0$. In particular, the spectral dimension of $(-\Delta+V)$ on $F_\infty$ is $d_s + 2(d_f/\beta)$, where $d_s$ is the spectral dimension of the bare Laplacian on $F$.
\end{proposition}
\begin{proof}
By \cite{Hamblyspec}*{Theorem 1.1} and \cite{Kajinospec}*{Theorem 1.2}, Assumption \ref{ass:WeylHKT} is satisfied on the constituent Sierpinski carpet $F$. In fact, \cite{Kajinospec2}*{Theorem 4.10} provides a sharper result of the form
\begin{equation}
\mathcal{L}^{\rm b}(F,\mu,t) = t^{-d_s/2} H(-\log t) + \sum_{k=1}^d t^{-d_k/d_w} G^{\rm b}_k(-\log t) + O\left(\exp(-c t^{-\frac{1}{d_w-1}})\right)
\end{equation}
as $t\downarrow 0$, where $H$ and the $G_k^{\rm b}$ are continuous periodic functions, $d_k$ is the Minkowski dimension of $F \cap \{x=(x_1,\cdots, x_d)\in \mathbb{R}^d: x_1=\cdots = x_k =0\}$, and $d_s$ and $d_w$ are respectively the spectral dimension and the walk dimension of $F$.

We turn our attention next to the potential term $\mathcal{F}^{\rm b}(V,t)$. It is direct to verify that for any $\beta>0$,
\begin{eqnarray}
 \int\limits_{K_\infty} e^{-t|x|^\beta} \, d\mu_\infty(x) &\leq& \int_0^\infty e^{-t\lambda} \, \frac{d\mu_\infty(\{x:  |x|^\beta<\lambda\})}{d\lambda} \,d\lambda \\
 &=& \int_0^\infty t e^{-t\lambda} \mu_\infty(B(0,\lambda^{1/\beta})) \, d\lambda \\
 &\leq& c_2 t \int_0^\infty e^{-t\lambda} \lambda^{d_f/\beta}\, d\lambda \leq C_2(d_f, \beta) t^{d_f/\beta},
\end{eqnarray}
and similarly
\begin{equation}
\int\limits_{K_\infty} e^{-t|x|^\beta} \, d\mu_\infty(x)  \geq C_1(d_f,\beta) t^{d_f/\beta}.
\end{equation}
Using the inequality $e^s \geq 1+s$ for $s \in \mathbb{R}$, we find
\begin{eqnarray*}
\left|e^{-t|x-y|^\beta}-e^{-t|x-z|^\beta}\right| 
&\leq&  \max\left(e^{-t|x-y|^\beta},e^{-t|x-z|^\beta} \right)\cdot  t\left(|x-y|^\beta-|x-z|^\beta\right)\\
&\leq&  C_\beta \cdot t  \cdot \max\left(e^{-t|x-y|^\beta},e^{-t|x-z|^\beta} \right)\cdot  |y-z|.
\end{eqnarray*}
It follows that as $t\downarrow 0$,
\begin{equation}
\mathcal{F}^\vee(V,t)-\mathcal{F}^\wedge(V,t) \leq C \cdot O(t \mathcal{F}(V,t)) = o(\mathcal{F}(V,t)),
\end{equation}
leading to the error estimate
\begin{equation}
\left|\frac{\mathcal{L}(F_\infty,\mu_\infty,V,t)}{t^{-d_s/2} H(-\log t) \mathcal{F}(V,t)}-1\right| 
= O\left(t^{(d_0-d_1)/d_w}\right)
\end{equation}
as $t\downarrow 0$. The Laplace transform version of Bohr's formula then follows.
\end{proof}


\subsubsection*{Acknowledgements}

We thank Luke Rogers for providing 
many 
 constructive comments regarding this work.

%




\begin{bibdiv}
\begin{biblist}

\bib{A}{article}{
   author={Akkermans, Eric},
   title={Statistical mechanics and quantum fields on fractals},
   conference={
      title={Fractal geometry and dynamical systems in pure and applied
      mathematics. II. Fractals in applied mathematics},
   },
   book={
      series={Contemp. Math.},
      volume={601},
      publisher={Amer. Math. Soc., Providence, RI},
   },
   date={2013},
   pages={1--21},
   review={\MR{3203824}},
   doi={10.1090/conm/601/11962},
}

\bib{s2}{article}{
   author={Allan, Adam},
   author={Barany, Michael},
   author={Strichartz, Robert S.},
   title={Spectral operators on the Sierpinski gasket. I},
   journal={Complex Var. Elliptic Equ.},
   volume={54},
   date={2009},
   number={6},
   pages={521--543},
   issn={1747-6933},
   review={\MR{2537254 (2010f:28011)}},
   doi={10.1080/17476930802272978},
}
		
		
\bib{w2}{article}{
   author={U. Andrews},
   author={G. Bonik},
   author={J. P. Chen},
   author={R. W. Martin},
   author={A. Teplyaev},
   title={Wave equation on one-dimensional fractals with spectral decimation and the complex dynamics of polynomials},
   journal={arXiv:1505.05855},
   date={2015},
}

\bib{eigen}{article}{
   author={Bajorin, N.},
   author={Chen, T.},
   author={Dagan, A.},
   author={Emmons, C.},
   author={Hussein, M.},
   author={Khalil, M.},
   author={Mody, P.},
   author={Steinhurst, B.},
   author={Teplyaev, A.},
   title={Vibration modes of $3n$-gaskets and other fractals},
   journal={J. Phys. A},
   volume={41},
   date={2008},
   number={1},
   pages={015101, 21},
   issn={1751-8113},
   review={\MR{2450694 (2010a:28008)}},
   doi={10.1088/1751-8113/41/1/015101},
}

\bib{BarlowStFlour}{article}{
   author={Barlow, Martin T.},
   title={Diffusions on fractals},
   conference={
      title={Lectures on probability theory and statistics},
      address={Saint-Flour},
      date={1995},
   },
   book={
      series={Lecture Notes in Math.},
      volume={1690},
      publisher={Springer, Berlin},
   },
   date={1998},
   pages={1--121},
   review={\MR{1668115 (2000a:60148)}},
   doi={10.1007/BFb0092537},
}

\bib{BarlowBass}{article}{
   author={Barlow, Martin T.},
   author={Bass, Richard F.},
   title={Brownian motion and harmonic analysis on Sierpinski carpets},
   journal={Canad. J. Math.},
   volume={51},
   date={1999},
   number={4},
   pages={673--744},
   issn={0008-414X},
   review={\MR{1701339 (2000i:60083)}},
   doi={10.4153/CJM-1999-031-4},
}

\bib{BBKT}{article}{
   author={Barlow, Martin T.},
   author={Bass, Richard F.},
   author={Kumagai, Takashi},
   author={Teplyaev, Alexander},
   title={Uniqueness of Brownian motion on Sierpi\'nski carpets},
   journal={J. Eur. Math. Soc. (JEMS)},
   volume={12},
   date={2010},
   number={3},
   pages={655--701},
   issn={1435-9855},
   review={\MR{2639315 (2011i:60146)}},
}

\bib{BP}{article}{
   author={Barlow, Martin T.},
   author={Perkins, Edwin A.},
   title={Brownian motion on the Sierpi\'nski gasket},
   journal={Probab. Theory Related Fields},
   volume={79},
   date={1988},
   number={4},
   pages={543--623},
   issn={0178-8051},
   review={\MR{966175 (89g:60241)}},
   doi={10.1007/BF00318785},
}


\bib{g3}{article}{
   author={Begue, Matthew},
   author={Kelleher, Daniel J.},
   author={Nelson, Aaron},
   author={Panzo, Hugo},
   author={Pellico, Ryan},
   author={Teplyaev, Alexander},
   title={Random walks on barycentric subdivisions and the Strichartz
   hexacarpet},
   journal={Exp. Math.},
   volume={21},
   date={2012},
   number={4},
   pages={402--417},
   issn={1058-6458},
   review={\MR{3004256}},
}
		
\bib{ph-b}{article}{
   author={Bellissard, J.},
   title={Renormalization group analysis and quasicrystals},
   conference={
      title={Ideas and methods in quantum and statistical physics},
      address={Oslo},
      date={1988},
   },
   book={
      publisher={Cambridge Univ. Press, Cambridge},
   },
   date={1992},
   pages={118--148},
   review={\MR{1190523 (93k:81045)}},
}

\bib{BST}{article}{
   author={Ben-Bassat, Oren},
   author={Strichartz, Robert S.},
   author={Teplyaev, Alexander},
   title={What is not in the domain of the Laplacian on Sierpinski gasket
   type fractals},
   journal={J. Funct. Anal.},
   volume={166},
   date={1999},
   number={2},
   pages={197--217},
   issn={0022-1236},
   review={\MR{1707752 (2001e:31016)}},
   doi={10.1006/jfan.1999.3431},
}


\bib{BrCa}{article}{
   author={Brossard, Jean},
   author={Carmona, Ren{\'e}},
   title={Can one hear the dimension of a fractal?},
   journal={Comm. Math. Phys.},
   volume={104},
   date={1986},
   number={1},
   pages={103--122},
   issn={0010-3616},
   review={\MR{834484 (87h:58218)}},
}




\bib{w3}{article}{
   author={J. F.-C. Chan},
   author={S.-M. Nga},
   author={A. Teplyaev},
   title={One-dimensional wave equations defined by fractal Laplacians},
   journal={to appear in 
   J. d'Analyse M.,
   arXiv:1406.0207},
   date={2015},
}

\bib{s0}{article}{
   author={Coletta, Kevin},
   author={Dias, Kealey},
   author={Strichartz, Robert S.},
   title={Numerical analysis on the Sierpinski gasket, with applications to
   Schr\"odinger equations, wave equation, and Gibbs' phenomenon},
   journal={Fractals},
   volume={12},
   date={2004},
   number={4},
   pages={413--449},
   issn={0218-348X},
   review={\MR{2109985 (2005k:65245)}},
   doi={10.1142/S0218348X04002689},
}



\bib{s1}{article}{
   author={Constantin, Sarah},
   author={Strichartz, Robert S.},
   author={Wheeler, Miles},
   title={Analysis of the Laplacian and spectral operators on the Vicsek
   set},
   journal={Commun. Pure Appl. Anal.},
   volume={10},
   date={2011},
   number={1},
   pages={1--44},
   issn={1534-0392},
   review={\MR{2746525 (2012b:28012)}},
   doi={10.3934/cpaa.2011.10.1},
}

\bib{dgv}{article}{
   author={Derfel, Gregory},
   author={Grabner, Peter J.},
   author={Vogl, Fritz},
   title={Laplace operators on fractals and related functional equations},
   journal={J. Phys. A},
   volume={45},
   date={2012},
   number={46},
   pages={463001, 34},
   issn={1751-8113},
   review={\MR{2993415}},
   doi={10.1088/1751-8113/45/46/463001},
}
		

\bib{g1}{article}{
   author={Drenning, Shawn},
   author={Strichartz, Robert S.},
   title={Spectral decimation on Hambly's homogeneous hierarchical gaskets},
   journal={Illinois J. Math.},
   volume={53},
   date={2009},
   number={3},
   pages={915--937 (2010)},
   issn={0019-2082},
   review={\MR{2727362 (2012b:28015)}},
}

\bib{Dunne}{article}{
   author={Dunne, Gerald V.},
   title={Heat kernels and zeta functions on fractals},
   journal={J. Phys. A},
   volume={45},
   date={2012},
   number={37},
   pages={374016, 22},
   issn={1751-8113},
   review={\MR{2970533}},
   doi={10.1088/1751-8113/45/37/374016},
}

\bib{StrichartzSHO}{article}{
   author={Fan, Edward},
   author={Khandker, Zuhair},
   author={Strichartz, Robert S.},
   title={Harmonic oscillators on infinite Sierpinski gaskets},
   journal={Comm. Math. Phys.},
   volume={287},
   date={2009},
   number={1},
   pages={351--382},
   issn={0010-3616},
   review={\MR{2480752 (2011f:35059)}},
   doi={10.1007/s00220-008-0633-z},
}

\bib{FellerVol2}{book}{
   author={Feller, William},
   title={An introduction to probability theory and its applications. Vol.
   II. },
   series={Second edition},
   publisher={John Wiley \& Sons, Inc., New York-London-Sydney},
   date={1971},
   pages={xxiv+669},
   review={\MR{0270403 (42 \#5292)}},
}

\bib{FlVa}{article}{
   author={Fleckinger-Pell{\'e}, Jacqueline},
   author={Vassiliev, Dmitri G.},
   title={An example of a two-term asymptotics for the ``counting function''
   of a fractal drum},
   journal={Trans. Amer. Math. Soc.},
   volume={337},
   date={1993},
   number={1},
   pages={99--116},
   issn={0002-9947},
   review={\MR{1176086 (93g:58147)}},
   doi={10.2307/2154311},
}

\bib{GriNad1}{article}{
   author={Grigor'yan, A.},
   author={Nadirashvili, N.},
   title={Negative eigenvalues of two-dimensional Schr{\"o}dinger operators},
   journal={arXiv:1112.4986},
   date={2014},
}



\bib{GriNad2}{article}{
   author={Grigor'yan, Alexander},
   author={Nadirashvili, Nikolai},
   title={Negative Eigenvalues of Two-Dimensional Schr\"odinger Operators},
   journal={Arch. Ration. Mech. Anal.},
   volume={217},
   date={2015},
   number={3},
   pages={975--1028},
   issn={0003-9527},
   review={\MR{3356993}},
   doi={10.1007/s00205-015-0848-z},
}

\bib{Hamblyspec}{article}{
   author={Hambly, B. M.},
   title={Asymptotics for functions associated with heat flow on the
   Sierpinski carpet},
   journal={Canad. J. Math.},
   volume={63},
   date={2011},
   number={1},
   pages={153--180},
   issn={0008-414X},
   review={\MR{2779136 (2012f:35548)}},
   doi={10.4153/CJM-2010-079-7},
}
	

\bib{HamKumFields}{article}{
   author={Hambly, B. M.},
   author={Kumagai, T.},
   title={Diffusion processes on fractal fields: heat kernel estimates and
   large deviations},
   journal={Probab. Theory Related Fields},
   volume={127},
   date={2003},
   number={3},
   pages={305--352},
   issn={0178-8051},
   review={\MR{2018919 (2004k:60219)}},
   doi={10.1007/s00440-003-0284-0},
}

\bib{Hare}{article}{
   author={Hare, Kathryn E.},
   author={Steinhurst, Benjamin A.},
   author={Teplyaev, Alexander},
   author={Zhou, Denglin},
   title={Disconnected Julia sets and gaps in the spectrum of Laplacians on
   symmetric finitely ramified fractals},
   journal={Math. Res. Lett.},
   volume={19},
   date={2012},
   number={3},
   pages={537--553},
   issn={1073-2780},
   review={\MR{2998138}},
   doi={10.4310/MRL.2012.v19.n3.a3},
}

\bib{ph1}{article}{
   author={Hinz, Michael},
   author={Teplyaev, Alexander},
   title={Dirac and magnetic Schr\"odinger operators on fractals},
   journal={J. Funct. Anal.},
   volume={265},
   date={2013},
   number={11},
   pages={2830--2854},
   issn={0022-1236},
   review={\MR{3096991}},
   doi={10.1016/j.jfa.2013.07.021},
}

\bib{HoltMolchanov}{article}{
   author={Holt, J.},
   author={Molchanov, S.},
   title={On the Bohr formula for the one-dimensional Schr{\"o}dinger operator
   with increasing potential},
   journal={Appl. Anal.},
   volume={84},
   date={2005},
   number={6},
   pages={555--569},
   issn={0003-6811},
   review={\MR{2151668 (2006g:34201)}},
   doi={10.1080/00036810500047899},
}



\bib{i1}{article}{
   author={Ionescu, Marius},
   author={Rogers, Luke G.},
   title={Complex powers of the Laplacian on affine nested fractals as
   Calder\'on-Zygmund operators},
   journal={Commun. Pure Appl. Anal.},
   volume={13},
   date={2014},
   number={6},
   pages={2155--2175},
   issn={1534-0392},
   review={\MR{3248383}},
   doi={10.3934/cpaa.2014.13.2155},
}
		

\bib{i2}{article}{
   author={Ionescu, Marius},
   author={Rogers, Luke G.},
   author={Strichartz, Robert S.},
   title={Pseudo-differential operators on fractals and other metric measure
   spaces},
   journal={Rev. Mat. Iberoam.},
   volume={29},
   date={2013},
   number={4},
   pages={1159--1190},
   issn={0213-2230},
   review={\MR{3148599}},
   doi={10.4171/RMI/752},
}
		


		


		

\bib{i3}{article}{
   author={Ionescu, Marius},
   author={Pearse, Erin P. J.},
   author={Rogers, Luke G.},
   author={Ruan, Huo-Jun},
   author={Strichartz, Robert S.},
   title={The resolvent kernel for PCF self-similar fractals},
   journal={Trans. Amer. Math. Soc.},
   volume={362},
   date={2010},
   number={8},
   pages={4451--4479},
   issn={0002-9947},
   review={\MR{2608413 (2011b:28017)}},
   doi={10.1090/S0002-9947-10-05098-1},
}
		



\bib{Iv}{article}{
   author={Ivri{\u\i}, V. Ja.},
   title={The second term of the spectral asymptotics for a Laplace-Beltrami
   operator on manifolds with boundary},
   language={Russian, English translation: Functional Anal. Appl. {\bf14} (1980), 
   98--106.},
   journal={Funktsional. Anal. i Prilozhen.},
   volume={14},
   date={1980},
   number={2},
   pages={25--34},
   issn={0374-1990},
   review={\MR{575202 (82m:58057)}},
}

\bib{KigamiLapidus}{article}{
   author={Kigami, Jun},
   author={Lapidus, Michel L.},
   title={Weyl's problem for the spectral distribution of Laplacians on
   p.c.f.\ self-similar fractals},
   journal={Comm. Math. Phys.},
   volume={158},
   date={1993},
   number={1},
   pages={93--125},
   issn={0010-3616},
   review={\MR{1243717 (94m:58225)}},
}

\bib{KS}{book}{
   author={Kostju{\v{c}}enko, A. G.},
   author={Sargsyan, I. S.},
   title={Raspredelenie sobstvennykh znachenii},
   language={Russian},
   note={Samosopryazhennye obyknovennye differentsialnye operatory.
   [Selfadjoint ordinary differential operators]},
   publisher={``Nauka'', Moscow},
   date={1979},
   pages={400},
   review={\MR{560900 (81j:34034)}},
}
	

\bib{LaPo}{article}{
   author={Lapidus, Michel L.},
   author={Pomerance, Carl},
   title={Counterexamples to the modified Weyl-Berry conjecture on fractal
   drums},
   journal={Math. Proc. Cambridge Philos. Soc.},
   volume={119},
   date={1996},
   number={1},
   pages={167--178},
   issn={0305-0041},
   review={\MR{1356166 (96h:58175)}},
   doi={10.1017/S0305004100074053},
}

\bib{LevitanSargsjan}{book}{
   author={Levitan, B. M.},
   author={Sargsyan, I. S.},
   title={Introduction to spectral theory: selfadjoint ordinary differential
   operators},
   note={Translated from the Russian by Amiel Feinstein;
   Translations of Mathematical Monographs, Vol. 39},
   publisher={American Mathematical Society, Providence, R.I.},
   date={1975},
   pages={xi+525},
   review={\MR{0369797 (51 \#6026)}},
}



\bib{NaimarkSolomyak}{article}{
   author={Naimark, K.},
   author={Solomyak, M.},
   title={Regular and pathological eigenvalue behavior for the equation
   $-\lambda u''=Vu$ on the semiaxis},
   journal={J. Funct. Anal.},
   volume={151},
   date={1997},
   number={2},
   pages={504--530},
   issn={0022-1236},
   review={\MR{1491550 (99b:34039)}},
   doi={10.1006/jfan.1997.3149},
}

\bib{Kajinospec}{article}{
   author={Kajino, Naotaka},
   title={Spectral asymptotics for Laplacians on self-similar sets},
   journal={J. Funct. Anal.},
   volume={258},
   date={2010},
   number={4},
   pages={1310--1360},
   issn={0022-1236},
   review={\MR{2565841 (2011j:31010)}},
   doi={10.1016/j.jfa.2009.11.001},
}

\bib{Kajinospec2}{article}{
   author={Kajino, Naotaka},
   title={Log-periodic asymptotic expansion of the spectral partition
   function for self-similar sets},
   journal={Comm. Math. Phys.},
   volume={328},
   date={2014},
   number={3},
   pages={1341--1370},
   issn={0010-3616},
   review={\MR{3201226}},
   doi={10.1007/s00220-014-1922-3},
}

\bib{Kigami}{book}{
   author={Kigami, Jun},
   title={Analysis on fractals},
   series={Cambridge Tracts in Mathematics},
   volume={143},
   publisher={Cambridge University Press, Cambridge},
   date={2001},
   pages={viii+226},
   isbn={0-521-79321-1},
   review={\MR{1840042 (2002c:28015)}},
   doi={10.1017/CBO9780511470943},
}

\bib{w4}{article}{
   author={Kusuoka, Shigeo},
   author={Zhou, Xian Yin},
   title={Waves on fractal-like manifolds and effective energy propagation},
   journal={Probab. Theory Related Fields},
   volume={110},
   date={1998},
   number={4},
   pages={473--495},
   issn={0178-8051},
   review={\MR{1626955 (99i:58148)}},
   doi={10.1007/s004400050156},
}

\bib{Lindstrom}{article}{
   author={Lindstr{\o}m, Tom},
   title={Brownian motion on nested fractals},
   journal={Mem. Amer. Math. Soc.},
   volume={83},
   date={1990},
   number={420},
   pages={iv+128},
   issn={0065-9266},
   review={\MR{988082 (90k:60157)}},
   doi={10.1090/memo/0420},
}

\bib{o1}{article}{
   author={Okoudjou, Kasso A.},
   author={Rogers, Luke G.},
   author={Strichartz, Robert S.},
   title={Szeg\"o limit theorems on the Sierpi\'nski gasket},
   journal={J. Fourier Anal. Appl.},
   volume={16},
   date={2010},
   number={3},
   pages={434--447},
   issn={1069-5869},
   review={\MR{2643590 (2011c:35380)}},
   doi={10.1007/s00041-009-9102-0},
}

\bib{OSt}{article}{
   author={Okoudjou, Kasso A.},
   author={Strichartz, Robert S.},
   title={Weak uncertainty principles on fractals},
   journal={J. Fourier Anal. Appl.},
   volume={11},
   date={2005},
   number={3},
   pages={315--331},
   issn={1069-5869},
   review={\MR{2167172 (2006f:28011)}},
   doi={10.1007/s00041-005-4032-y},
}

\bib{OS-CT}{article}{
   author={Okoudjou, Kasso A.},
   author={Saloff-Coste, Laurent},
   author={Teplyaev, Alexander},
   title={Weak uncertainty principle for fractals, graphs and metric measure
   spaces},
   journal={Trans. Amer. Math. Soc.},
   volume={360},
   date={2008},
   number={7},
   pages={3857--3873},
   issn={0002-9947},
   review={\MR{2386249 (2008k:42121)}},
   doi={10.1090/S0002-9947-08-04472-3},
}

\bib{Qiu}{article}{
	author={Qiu, Hua},
	title={Exact spectrum of the Laplacian on a domain in the Sierpinski gasket},
	journal={arXiv:1206.1381v2},
	date={2012},
}

\bib{q}{article}{
   author={Quint, J.-F.},
   title={Harmonic analysis on the Pascal graph},
   journal={J. Funct. Anal.},
   volume={256},
   date={2009},
   number={10},
   pages={3409--3460},
   issn={0022-1236},
   review={\MR{2504530 (2010e:37053)}},
   doi={10.1016/j.jfa.2009.01.011},
}

\bib{ReedSimonVol4}{book}{
   author={Reed, Michael},
   author={Simon, Barry},
   title={Methods of modern mathematical physics. IV. Analysis of operators},
   publisher={Academic Press [Harcourt Brace Jovanovich, Publishers], New
   York-London},
   date={1978},
   pages={xv+396},
   isbn={0-12-585004-2},
   review={\MR{0493421 (58 \#12429c)}},
}

\bib{r1}{article}{
   author={Rogers, Luke G.},
   title={Estimates for the resolvent kernel of the Laplacian on p.c.f.
   self-similar fractals and blowups},
   journal={Trans. Amer. Math. Soc.},
   volume={364},
   date={2012},
   number={3},
   pages={1633--1685},
   issn={0002-9947},
   review={\MR{2869187}},
   doi={10.1090/S0002-9947-2011-05551-0},
}

\bib{RT}{article}{
   author={Rogers, Luke G.},
   author={Teplyaev, Alexander},
   title={Laplacians on the basilica Julia sets},
   journal={Commun. Pure Appl. Anal.},
   volume={9},
   date={2010},
   number={1},
   pages={211--231},
   issn={1534-0392},
   review={\MR{2556753 (2011c:28024)}},
   doi={10.3934/cpaa.2010.9.211},
}

\bib{RST}{article}{
   author={Rogers, Luke G.},
   author={Strichartz, Robert S.},
   author={Teplyaev, Alexander},
   title={Smooth bumps, a Borel theorem and partitions of smooth functions
   on P.C.F.\ fractals},
   journal={Trans. Amer. Math. Soc.},
   volume={361},
   date={2009},
   number={4},
   pages={1765--1790},
   issn={0002-9947},
   review={\MR{2465816 (2010f:28020)}},
   doi={10.1090/S0002-9947-08-04772-7},
}

\bib{RozenblumSolomyak}{article}{
   author={Rozenblum, G.},
   author={Solomyak, M.},
   title={On spectral estimates for the Schr\"odinger operators in global
   dimension 2},
   journal={Algebra i Analiz},
   volume={25},
   date={2013},
   number={3},
   pages={185--199},
   issn={0234-0852},
   translation={
      journal={St. Petersburg Math. J.},
      volume={25},
      date={2014},
      number={3},
      pages={495--505},
      issn={1061-0022},
   },
   review={\MR{3184603}},
   doi={10.1090/S1061-0022-2014-01301-5},
}

\bib{Shargorodsky}{article}{
   author={Shargorodsky, Eugene},
   title={On negative eigenvalues of two-dimensional Schr\"odinger
   operators},
   journal={Proc. Lond. Math. Soc. (3)},
   volume={108},
   date={2014},
   number={2},
   pages={441--483},
   issn={0024-6115},
   review={\MR{3166359}},
   doi={10.1112/plms/pdt036},
}

\bib{Ben}{article}{
   author={Steinhurst, Benjamin A.},
   author={Teplyaev, Alexander},
   title={Existence of a meromorphic extension of spectral zeta functions on
   fractals},
   journal={Lett. Math. Phys.},
   volume={103},
   date={2013},
   number={12},
   pages={1377--1388},
   issn={0377-9017},
   review={\MR{3117253}},
   doi={10.1007/s11005-013-0649-y},
}

\bib{StrFractafold}{article}{
   author={Strichartz, Robert S.},
   title={Fractafolds based on the Sierpi\'nski gasket and their spectra},
   journal={Trans. Amer. Math. Soc.},
   volume={355},
   date={2003},
   number={10},
   pages={4019--4043 (electronic)},
   issn={0002-9947},
   review={\MR{1990573 (2004b:28013)}},
   doi={10.1090/S0002-9947-03-03171-4},
}

\bib{StrichartzFractalsInTheLarge}{article}{
   author={Strichartz, Robert S.},
   title={Fractals in the large},
   journal={Canad. J. Math.},
   volume={50},
   date={1998},
   number={3},
   pages={638--657},
   issn={0008-414X},
   review={\MR{1629847 (99f:28015)}},
   doi={10.4153/CJM-1998-036-5},
}

\bib{g2}{article}{
   author={Strichartz, Robert S.},
   title={Laplacians on fractals with spectral gaps have nicer Fourier
   series},
   journal={Math. Res. Lett.},
   volume={12},
   date={2005},
   number={2-3},
   pages={269--274},
   issn={1073-2780},
   review={\MR{2150883 (2006e:28013)}},
   doi={10.4310/MRL.2005.v12.n2.a12},
}

\bib{StrBook}{book}{
   author={Strichartz, Robert S.},
   title={Differential equations on fractals},
   note={A tutorial},
   publisher={Princeton University Press, Princeton, NJ},
   date={2006},
   pages={xvi+169},
   isbn={978-0-691-12731-6},
   isbn={0-691-12731-X},
   review={\MR{2246975 (2007f:35003)}},
}


\bib{w1}{article}{
   author={Strichartz, Robert S.},
   title={Waves are recurrent on noncompact fractals},
   journal={J. Fourier Anal. Appl.},
   volume={16},
   date={2010},
   number={1},
   pages={148--154},
   issn={1069-5869},
   review={\MR{2587585 (2011a:35540)}},
   doi={10.1007/s00041-009-9103-z},
}

\bib{StrTep}{article}{
   author={Strichartz, Robert S.},
   author={Teplyaev, Alexander},
   title={Spectral analysis on infinite Sierpi\'nski fractafolds},
   journal={J. Anal. Math.},
   volume={116},
   date={2012},
   pages={255--297},
   issn={0021-7670},
   review={\MR{2892621}},
   doi={10.1007/s11854-012-0007-5},
}

\bib{T}{article}{
   author={Teplyaev, Alexander},
   title={Spectral analysis on infinite Sierpi\'nski gaskets},
   journal={J. Funct. Anal.},
   volume={159},
   date={1998},
   number={2},
   pages={537--567},
   issn={0022-1236},
   review={\MR{1658094 (99j:35153)}},
   doi={10.1006/jfan.1998.3297},
}

\bib{Tzeta}{article}{
   author={Teplyaev, Alexander},
   title={Spectral zeta functions of fractals and the complex dynamics of
   polynomials},
   journal={Trans. Amer. Math. Soc.},
   volume={359},
   date={2007},
   number={9},
   pages={4339--4358 (electronic)},
   issn={0002-9947},
   review={\MR{2309188 (2008j:11119)}},
   doi={10.1090/S0002-9947-07-04150-5},
}





\end{biblist}
\end{bibdiv}
\end{document}